%% file: main.tex
\documentclass[conference]{IEEEtran}
\IEEEoverridecommandlockouts
\usepackage{prologue, graphicx}
\usepackage{mdwlist}
\usepackage{epstopdf}
\usepackage{bm}
\usepackage{makecell}
\usepackage{cite}
\usepackage{color}
\usepackage{multirow}
\usepackage{enumitem}
\usepackage{bbm}
\usepackage{enumitem}
\def\BibTeX{{\rm B\kern-.05em{\sc i\kern-.025em b}\kern-.08em
    T\kern-.1667em\lower.7ex\hbox{E}\kern-.125emX}}
\begin{document}

\title{Stateful Switch: Optimized Time Series Release with Local Differential Privacy}

\author{
	Qingqing Ye $^{\dagger}$, Haibo Hu $^{\dagger}$, Kai Huang $^{\ddagger}$, Man Ho Au $^{\sharp}$, Qiao Xue $^{\dagger}$   \vspace{0.05in} \\
	
	$^{\dagger}$ The Hong Kong Polytechnic University, \\
	$^{\ddagger}$ The Hong Kong University of Science and Technology, 
	$^{\sharp}$ The University of Hong Kong
	 \vspace{0.05in} \\
	
	\emph{
		\{qqing.ye, haibo.hu, qiao2.xue\}@polyu.edu.hk,
		ustkhuang@ust.hk,
		allenau@cs.hku.hk
	}
}

\maketitle

\begin{abstract}
	Time series data have numerous applications in big data analytics. However, they often cause privacy issues when collected from individuals. To address this problem, most existing works perturb the values in the time series while retaining their temporal order, which may lead to significant distortion of the values. Recently, we propose TLDP model~\cite{ye2021beyond} that perturbs temporal perturbation to ensure privacy guarantee while retaining original values. It has shown great promise to achieve significantly higher utility than value perturbation mechanisms in many time series analysis. However, its practicability is still undermined by two factors, namely, utility cost of extra missing or empty values, and inflexibility of privacy budget settings. To address them, in this paper we propose {\it switch} as a new two-way operation for temporal perturbation, as opposed to the one-way {\it dispatch} operation in ~\cite{ye2021beyond}. The former inherently eliminates the cost of missing, empty or repeated values. Optimizing switch operation in a {\it stateful} manner,  we then propose $StaSwitch$ mechanism for time series release under TLDP. Through both analytical and empirical studies, we show that $StaSwitch$ has significantly higher utility for the published time series than any state-of-the-art temporal- or value-perturbation mechanism, while allowing any combination of privacy budget settings.
\end{abstract}


\section{Introduction}
\label{sec:introduction}
\input{introduction}

\section{Problem Definition and Preliminaries}
\label{sec:problem}
\input{problem}

\section{Switch Operation and RanSwitch Mechanism}
\label{sec:baseline}
\input{baseline}

\section{StaSwitch: A Mechanism with Stateful Switch}
\label{sec:staswitch}
\input{staswitch}

\section{Experimental Evaluation}
\label{sec:experiment}
\input{experiment}

\section{Related Work}
\label{sec:related_work}
\input{related}

\section{Conclusion}
\label{sec:conclusion}
This paper studies the problem of time series release following TLDP privacy model. We first define {\it switch} as a two-way atomic operation for the time series perturbation, which inherently eliminates missing, empty or repeated values. Then we propose a baseline mechanism $RanSwitch$ and an optimized mechanism $StaSwitch$, the latter of which adopts stateful switch to bound each value's timestamp deviation, and thus enhances the utility significantly. We compare $RanSwitch$ and $StaSwitch$ with the existing temporal-perturbation and value-perturbation mechanisms through extensive analytical and empirical analysis under various privacy budgets and time window sizes, and show that the optimized mechanism $StaSwitch$ always achieves the best performance in various tasks. 

As for future work, we plan to extend this work to more complicated time series analysis tasks, such as temporally correlated time series release, time series forecasting, pattern recognition and curve fitting.

\section*{Acknowledgment}
This work was supported by the National Natural Science Foundation of China (Grant No: 62102334, 62072390, 92270123 and 61972332), and the Research Grants Council, Hong Kong SAR, China (Grant No:  15222118, 15218919, 15203120, 15226221, 15225921, 15209922 and C2004-21GF).

\bibliographystyle{abbrv}
\bibliography{ref}

\end{document}

%% file: introduction.tex
In big data era, continual data, i.e., a sequence of values in the temporal order (a.k.a., \emph{time series}), has numerous real-world applications~\cite{wei2006time}. Among them, many time series are collected from individuals, such as biosensors in telecare, IoT sensors in smart home, and trajectories for mobility tracking in COVID-19 pandemic. Directly releasing them to the public can cause privacy infringement~\cite{papadimitriou2007time,huang2022privacy}. For example, the periodic heart rate readings from an Apple Watch may reveal the daily activity of its owner, e.g., sleeping, sitting, or walking.

To address this issue, many privacy-preserving time series publishing techniques have been proposed~\cite{shi2011privacy,yang2013towards}, most of which are based on differential privacy~\cite{dwork2006calibrating}, in either a centralized~\cite{dwork2010differential2, rastogi2010differentially} or a local setting~\cite{ding2017collecting, wang2021continuous, bao2021cgm, xue2022ddrm}. However, all these works are {\bf value perturbation} mechanisms, i.e., they perturb the value at each timestamp so that no value at any timestamp can be inferred with high confidence. Unfortunately, in medical and financial applications, such mechanisms do not work because distorted values are useless or even harmful, for example ECG/blood pressure readings, and stock trading prices.

Our recent work~\cite{ye2021beyond} has proposed to mitigate this issue by perturbing the {\bf temporal order} of a time series. The privacy model, namely local differential privacy in the temporal setting (TLDP), guarantees an adversary cannot infer the original timestamp of a value with high confidence. As temporal perturbation does not inject any noise to the value, the accuracy of most time series statistics (e.g., moving average, range count) and manipulations (e.g., window smoothing, resampling) can be significantly enhanced. The following is a concrete example beyond the medical or financial domain.

{\bf Example: Smart Meter.} Utility (electricity, gas, and water) companies are deploying smart meters in households to collect real-time consumption data for usage prediction and resource scheduling. However, such data may disclose the activities in an individual household, such as away-from-home (low usage of all three utilities) and heavy washing (high usage of both electricity and water in a laundry room). To preserve privacy under differential privacy, unfortunately we cannot perturb these reading values as they must be accurately reflected in the utility bills. Therefore, temporal perturbation (independently on these three time series) becomes the natural way to achieve deniability and differential privacy.

%On the other hand, many time series manipulations, such as aggregation (e.g., moving average, range count), window transformation (e.g., smoothing, rolling), and resampling in Apache Flink~\cite{Flink}, are operated on time windows, temporal perturbation can have less or even no impact on the accuracy of these operators than value perturbation.

Although TLDP is a promising privacy model in value-critical applications, there remain two issues in~\cite{ye2021beyond}. First, the proposed Threshold Mechanism (TM) is built on the operation of {\it dispatch}, which randomly moves the value of the current timestamp to a future timestamp within a sliding window of length $k$. But since this is a one-way operation (i.e., only from current to future, but not vice versa), it causes missing, empty or repeated values in a released time series. %Specifically, dispatching a value $S_i$ from timestamp $t_i$ to $t_j$ only means $S_i$ goes to the timestamp $t_j$, i.e., $S_i\Rightarrow t_j$, but without indicating where the value at timestamp $t_i$ is from. This asynchronism may cause dispatch conflict, which further results in utility cost.
Second, TLDP has two privacy parameters, namely, the privacy budget $\epsilon$ and the sliding window length $k$. However, the Threshold Mechanism (TM) cannot effectively support all combinations of $\epsilon$ and $k$. A mismatch of $\epsilon$ and $k$ could either cost TM extra missing or empty values to satisfy a small $\epsilon$ or a large $k$, or waste the large $\epsilon$ for a small $k$. The following two examples explain this mismatch issue of TM.
\begin{itemize}%[leftmargin=*]
	\item
	Frequency counting, which counts the occurrences of a specific value in a time series, is sensitive to missing values. Using TM with $\epsilon=1.0$ and $k=10$ causes very large estimation error, $382$ times higher than Randomized Response~\cite{warner1965randomized}, a value-perturbation mechanism on time series.
	
	\item
	When $k=4$, the largest privacy budget TM can support is only $2.19$. In other words, any $\epsilon$ larger than $2.19$ has to be wasted~\cite{ye2021beyond}.
\end{itemize}

In this paper, we present a switch-based mechanism for TLDP that addresses these two issues. As opposed to the dispatch operation in TM, {\it switch} is a two-way operation in a time series, which exchanges two values $S_i$ and $S_j$ of timestamps $t_i$ and $t_j$. In essence, a switch is equivalent to two synchronous dispatch operations, and it is free of missing, empty or repeated values. Based on this operation, we propose $StaSwitch$ (short for Stateful Switch) perturbation mechanism, which bounds each value's choice of switch by a stateful probability distribution. Furthermore, this new mechanism does not cause any mismatch on $\epsilon$ and $k$ as in TM. As such, users have full flexibility on the choice privacy parameters without degrading data utility. To summarize, our contributions in this paper are three-fold.
\begin{itemize}
\item
We propose a two-way atomic operation {\it switch} for temporal perturbation, which inherently eliminates missing, empty or repeated values in the released time series.

\item
We design two temporal perturbation mechanisms based on {\it switch} operation, namely the baseline mechanism $RanSwitch$ and an optimized one $StaSwitch$. They are capable of offering full flexibility for users to set any privacy parameters without degrading data utility.

\item
We present detailed analysis on the privacy guarantee and utility cost of $RanSwitch$ and $StaSwitch$. Through intensive analytical and empirical studies, we show that $StaSwitch$ has significantly higher utility for the released time series than any state-of-the-art temporal-perturbation or value-perturbation mechanism.
\end{itemize}

The rest of the paper is organized as follows. Section~\ref{sec:problem} formulates the problem of time series release. Section~\ref{sec:baseline} presents {\it switch} operation, together with our baseline mechanism $RanSwitch$.  Section~\ref{sec:staswitch} introduces $StaSwitch$ mechanism with theoretical analysis on privacy guarantee and utility cost. Section~\ref{sec:experiment} presents experimental results and case studies on both real and synthetic datasets. Finally, we review existing work in Section~\ref{sec:related_work} and conclude this paper in Section~\ref{sec:conclusion}. %Due to the space limitation, we provide the proof of theorems and a corollary, and omit the proof of lemmas in the paper.

%% file: problem.tex
\subsection{Problem Definition}
\label{sec:definition}
In this paper, we define a time series as an infinite sequence of values $S=\{S_1, S_2, ..., S_n, ...\}$ in a discrete temporal domain $T=\{t_1, t_2, ..., t_n, ...\}$. %with equal time intervals.
Our task is to release a sanitized time series $R=\{R_1, R_2, ..., R_n, ...\}$ out of the original one $S$ under local differential privacy, and as with~\cite{ye2021beyond}, our goal is to minimize the collective cost arising from each value's missing, repetition, empty and misaligned between $R$ and $S$. Specifically, a {\bf missing cost}, whose unit is $M$, occurs when a value in $S$ is missed in $R$; a {\bf repetition cost}, whose unit is $N$, occurs when a value is duplicated once in $S$; an {\bf empty cost}, whose unit is $E$, occurs when a timestamp in $R$ has not been filled with any value, causing a default; and finally a {\bf misalignment cost}~\footnote{The delay cost in~\cite{ye2021beyond} is a special case of misalignment cost. The latter also considers the cost when a value is released in advance of the original timestamp.} occurs when a value is released at an earlier or delayed timestamp, and one timestamp of misalignment bears a unit cost of $D$.

%In addition to minimizing the aforementioned costs while achieving TLDP as defined in Section~\ref{sec:TLDP_definition}, a perturbation mechanism for time series must also satisfy the following two practical constraints in time series analysis.
%\begin{enumerate}
%	\item
%	{\bf Online Mechanism.} The mechanism should be able to handle time series in the form of continual data streams, because many time series (e.g., sensor readings, and video feeds) are generated on the fly instead of released as a whole. Since the input stream is one value at a time, the mechanism must also release one value at a time.
%	
%	\item
%	{\bf Finite Memory.} The mechanism has limited memory and therefore cannot accommodate the entirety of an (infinite) time series. %Formally, we assume the mechanism can store up to $k$ values. %It is  noteworthy that such a limitation of memory also gives the mechanism an upper bound on the longest delay a value can be perturbed, which is essential in many real-time applications, such as autonomous driving, and financial trading.
%	
%\end{enumerate}

\subsection{Existing Value-Perturbation LDP Mechanisms for Time Series}
%A number of solutions have been proposed for time series release under LDP. Depending on the privacy requirements, canonical definitions of neighboring time series include {\it user-level} privacy~\cite{rastogi2010differentially} which allows two time series to be different for all values, {\it event-level}~\cite{dwork2010differential2} privacy which dictates two time series can only differ in one value, and {\it w-event} privacy~\cite{kellaris2014differentially} which allows two time series to differ in values within $w$ successive timestamps. Given a definition of neighboring time series as above, a formal definition of $(\epsilon, \delta)$-LDP is as below.
A number of solutions have been proposed for time series release under LDP. Depending on the privacy requirements, canonical definitions of neighboring time series include {\it user-level}~\cite{rastogi2010differentially}, {\it event-level}~\cite{dwork2010differential2}, and {\it w-event}~\cite{kellaris2014differentially} privacy. Given a definition of neighboring time series as above, a formal definition of $(\epsilon, \delta)$-LDP on time series is as below.
\begin{definition}
	({\bf $\bm{(\epsilon, \delta)}$-LDP})
%	({\bf Local Differential Privacy})
	Given privacy parameters $\epsilon$ and $\delta$, a randomized algorithm $\mathcal{A}$ satisfies $(\epsilon, \delta)$-LDP, iff for any two neighboring time series $S$ and $S^\prime$, and any possible output $R$ of $\mathcal{A}$, the following inequality holds:
	\begin{equation}
		\mathrm{Pr}(\mathcal{A}(S)= R) \le e^\epsilon \cdot \mathrm{Pr}(\mathcal{A}(S^\prime) = R) + \delta
	\end{equation}
\end{definition}

Since $(\epsilon, \delta)$-LDP is defined on difference in values, all existing works~\cite{erlingsson2014rappor, ding2017collecting, wang2021continuous, bao2021cgm} adopt {\it value-perturbation} mechanisms, such as Laplace mechanism~\cite{dwork2006calibrating}, Gaussian mechanism~\cite{dwork2014algorithmic}, and Randomized Response~\cite{warner1965randomized}, to inject noise to released value or statistics of the time series.

\subsection{Local Differential Privacy in Temporal Setting}
\label{sec:TLDP_definition}
As opposed to the above {\it value-perturbation} LDP model, we follow the {\it temporal-perturbation} LDP model (TLDP) as in~\cite{ye2021beyond}. In TLDP, two (temporally) neighboring time series are defined as those can be turned into one another by exchanging the values of two timestamps.

\begin{definition}
	\label{def:nighboring}
	({\bf Neighboring Time Series}) Two time series $S$ and $S'$ are neighbors if there exist two timestamps $t_i\neq t_j$ such that
	
	1) $|i-j| < k$, and
	
	2) $S_i=S'_j$ and $S_j=S'_i$, and
	
	3) for any other timestamp $t_l(l\neq i, j)$, $S_l=S'_l$.%\footnote{We assume $S$ and $S'$ are infinite time series. For finite time series, we treat the beginning and ending $k$ timestamps as warmup and cooldown periods and exclude them from this equality test.}
\end{definition}

In the above definition, $k$ is the length of a {\bf time sliding window}, which is an additional privacy parameter. The larger the $k$, the longer period the value remains sensitive to the user. For example, by setting $k$ to $24$ hours, a smart watch user can be assured that a released heart rate reading can be from anytime of that day; but if $k$ is set to $1$ hour, this period of ``deniability" is shortened to $1$ hour and might not be sufficient to preserve the user's privacy. Based on Definition~\ref{def:nighboring}, local differential privacy in the temporal setting, a.k.a. $(\epsilon, \delta)$-TLDP, is defined as follows.

\begin{definition}
	({\bf $\bm{(\epsilon, \delta)}$-TLDP})
%	({\bf Local Differential Privacy in the Temporal Setting})
	Given privacy parameters $\epsilon$ and $\delta$, a randomized algorithm $\mathcal{A}$ satisfies $(\epsilon, \delta)$-TLDP, iff for any two neighboring (in a window of length $k$) time series $S$ and $S^\prime$, and any possible output $R$ of $\mathcal{A}$, the following inequality holds:
	\begin{equation}
		\mathrm{Pr}(\mathcal{A}(S)= R) \le e^\epsilon \cdot 	\mathrm{Pr}(\mathcal{A}(S^\prime) = R) + \delta
	\end{equation}
\end{definition}

The degree of privacy in TLDP is controlled by $\epsilon$, $\delta$, and $k$. %Since the whole time series is released for analysis, the output $R$ of $\mathcal{A}$ is simply a perturbed time series from the original one.

%% file: baseline.tex
In this section, we first define the {\it switch} operation for temporal perturbation, based on which we present a baseline mechanism $RanSwitch$ to satisfy  $(\epsilon, \delta)$-TLDP, together with its perturbation protocol, and privacy and utility analysis.

\subsection{Switch Operation}
To perturb a time series temporally, an intuitive operation is to probabilistically assign a temporal position for the incoming value at each timestamp. This is the rationale of the {\it dispatch} operation in~\cite{ye2021beyond}. However, since the dispatch position is independently selected at each timestamp, dispatch conflicts may occur. We argue that the root cause lies in the {\bf one-way nature} of dispatch operation. For example, at timestamp $t_i$, dispatching its value $S_i$ to timestamp $t_j$ only decides the destination of $S_i$ is $t_j$, but it is uncertain which value $S_x$ should fill in $t_i$. In other words, the timestamp $t_i$'s ``from'' and ``to'' dispatches (i.e., $S_x \Rightarrow t_i$,  and $S_i \Rightarrow t_j$) always happen asynchronously and independently, which leads to missing, empty and repeated values.\footnote{According to~\cite{ye2021beyond}, if the ``from'' dispatch fails with conflict, then $t_i$ has to report an empty value; if the ``to'' dispatch fails, then the value $S_i$ will be missed; if two or more values happen to be dispatched to the same timestamp, or if a value is repeatedly dispatched to more than one timestamp, some values will be overwritten and thus missed.} To address this issue, in this paper we propose {\it switch} as a {\bf two-way atomic operation} for temporal perturbation, which is formally defined below.
\begin{definition}
	({\bf Switch Operation}) Given a sliding window of length $k$, a switch operation $S_i \Leftrightarrow S_j$ exchanges two values $S_i$ and $S_j$ with each other in $R$, that is, $R_j = S_i$ and $R_i = S_j$, where $0\le i-j<k$.
\end{definition}

In essence, a switch operation $S_i \Leftrightarrow S_j$ is equivalent to two simultaneous and correlated dispatch operations $S_i\Rightarrow t_j$ (i.e., ``to" dispatch at $t_i$) and $S_j\Rightarrow t_i$ (i.e., ``from" dispatch at $t_i$) in~\cite{ye2021beyond}. % By applying {\it switch} operation to any two values $S_i$ and $S_j$, actions ``from'' and ``to'' at timestamp $t_i$ become simultaneous and deterministic.
Therefore, it inherently eliminates missing, empty or repeated values. %As such, a switch operation actually dispatches two values from their own timestamps to others.
In the sequel, to avoid confusion with ``dispatch'', we use term ``{\bf allocate}" for the one-way perturbation in a switch operation, i.e., $S_i$ is allocated to $t_j$ and $S_j$ is allocated to $t_i$.

Another advantage of the switch operation is its inherit resemblance to neighboring time series in Definition~\ref{def:nighboring}. Recall that two temporally neighboring time series are those which can be turned into one another by exchanging the values of two timestamps. Therefore, to satisfy $(\epsilon, \delta)$-TLDP becomes intuitive --- at each timestamp we just randomly switch the current value with another one within the sliding window of length $k$. This idea leads to our baseline perturbation mechanism $RanSwitch$ for TLDP. In what follows, we will first present the perturbation protocol of $RanSwitch$ in Sec.~\ref{sec:baseline_protocol}, and then analyze its privacy guarantee in Sec.~\ref{sec:baseline_privacy} and utility cost in Sec.~\ref{sec:baseline_cost}.

\subsection{RanSwitch: A Baseline Mechanism}
\label{sec:baseline_protocol}
For a time series $S = \{S_1, S_2, ..., S_n, ...\}$, at each timestamp $t_i$, $RanSwitch$ randomly selects a timestamp $t_j$ in the sliding window $\{t_i, t_{i+1}, ..., t_{i+k-1}\}$ according to the following {\bf perturbation probability distribution}:
\begin{align}
	\label{eq:switch_naive}
	\mathrm{Pr}[j=i+l] = \begin{cases}
		p, \quad \text{if} \quad l=0 \\
		q, \quad \text{if} \quad l \in \{1, 2, ..., k-1\}
	\end{cases}
\end{align}
and then applies the switch operation to values $S_i$ and $S_j$. Here $p$ denotes the probability of selecting the current timestamp (i.e., retaining $S_i$), while $q$ is the probability of selecting one of the other $k-1$ timestamps, and $p+(k-1)q=1$. %For simplicity, we assume $p>q$ for the sake of data utility.
The pseudo-code of $RanSwitch$ mechanism is shown in Algorithm~\ref{alg:ranswitch}. $S_i$ may be finally allocated to one of the $k-1$ {\bf backward} timestamps $\{t_{i-k+1}, t_{i-k+2}, ..., t_{i-1}\}$, the {\bf current} timestamp $t_i$, or one of the {\bf forward} timestamps $\{t_{i+1}, t_{i+2}, ...\}$.

\vspace{-0.05in}
 \begin{algorithm}
 	\small
 	\caption{ Perturbation protocol of $RanSwitch$ }
 	\begin{tabular}{ll}
 		{\bf Input:} & Original time series $S = \{S_1, S_2, ..., S_n, ...\}$\\
 		& Sliding window length $k$ \\
 		& Perturbation probabilities $p$ and $q $ \\
 		{\bf Output:} & Released time series $R = \{R_1, R_2, ..., R_n, ...\}$ \\
% 		{\bf Procedure:} &
 	\end{tabular}
 	\label{alg:ranswitch}
 	\begin{algorithmic}[1]
 		\STATE Initialize the released time series $R=\emptyset$
 		\FOR{each timestamp $t_i$ ($i \in \{1, 2, ..., n, ...\}$)}
 		\STATE Randomly select an index $j$ from $\mathbb{X} = \{l|i \le l \le i+k-1\}$ according to Eq.~\ref{eq:switch_naive}
 		\STATE Switch $S_i$ and $S_j$
 		\STATE Set $R_i=S_i$ and release $R_i$
 		\ENDFOR
 		%\RETURN $R = \{R_1, R_2, ..., R_n, ...\}$
 	\end{algorithmic}
 \end{algorithm}

\begin{figure}[h]
	\vspace{-0.19in}
	\centering
	\subfigure[Switch operation]{
		\includegraphics[width=0.38\linewidth]{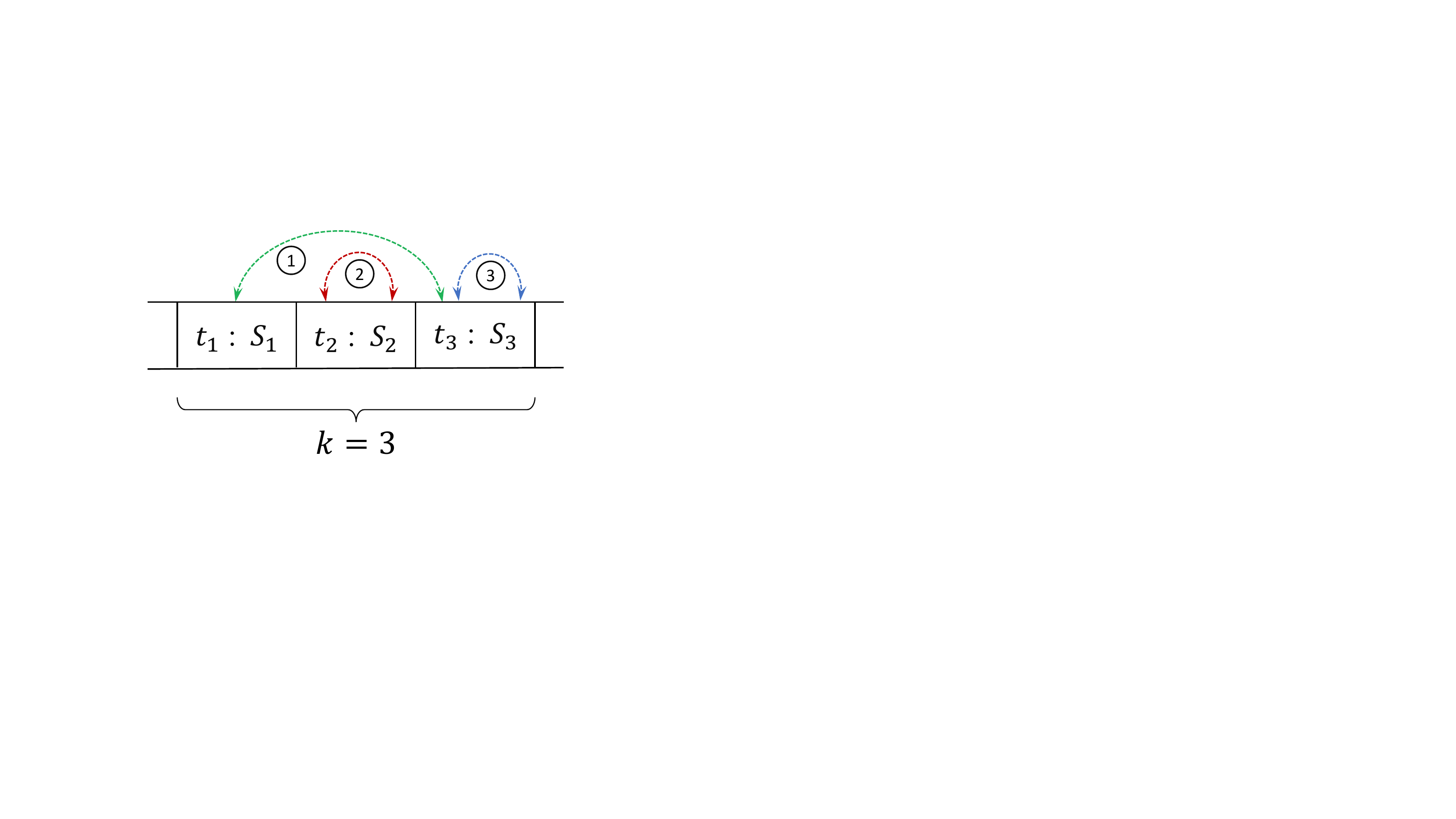}
	}
	\subfigure[Probability distribution]{
		\includegraphics[width=0.4\linewidth]{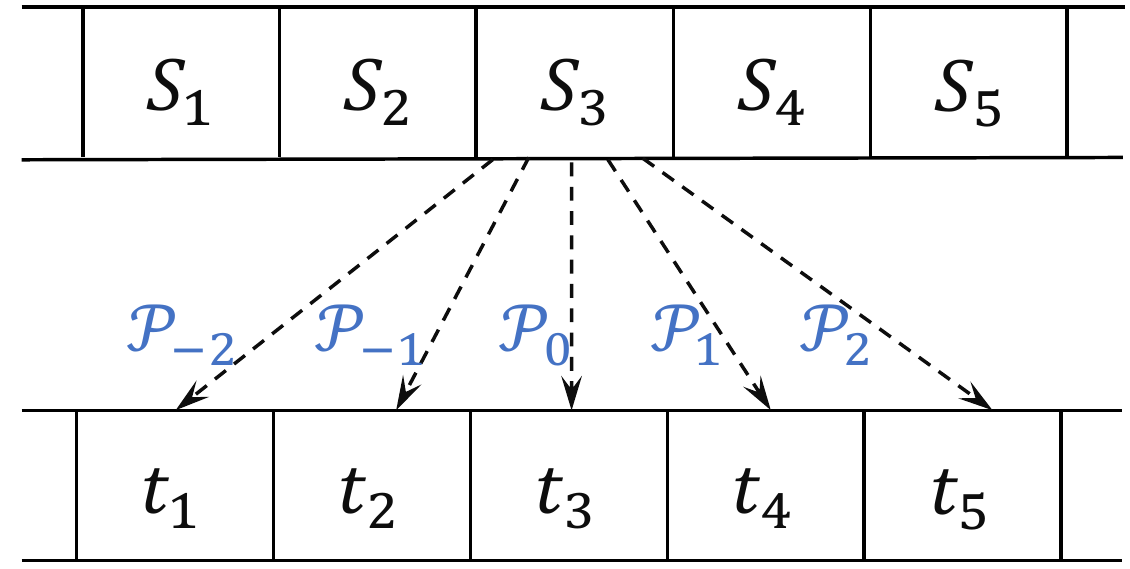}
	}
	\vspace{-0.09in}
	\caption{Dispatching probability of $RanSwitch$ mechanism}
	\vspace{-0.06in}
	\label{fig:switch_example}
\end{figure}

We show an example of $k=3$ in Fig.~\ref{fig:switch_example}(a). Since there are $3$ timestamps, $RanSwitch$ performs $3$ switch operations \textcircled{\footnotesize 1}\textcircled{\footnotesize 2}\textcircled{\footnotesize 3}, one at each timestamp. At timestamp $t_1$, suppose $S_1$ switches with $S_3$, then $S_3$ is allocated backward to $t_1$ and released. %\footnote{Because $S_3$ cannot be accessed by $RanSwitch$ at the next timestamp $t_2$.} (via switch operation  \textcircled{\footnotesize 1}). As such, $S_3$ is allocated to its backward timestamp. 
At timestamp $t_2$, suppose $t_2$ itself is selected, then $S_2$ is released at the current timestamp.% (via switch operation \textcircled{\footnotesize 2}). 
At timestamp $t_3$, suppose $t_3$ itself is selected, then the original $S_1$ (denoted by $S_3$ now) is allocated to its forward timestamp $t_3$ and then released.% (via switch operations \textcircled{\footnotesize 1} and \textcircled{\footnotesize 3}).

Since $RanSwitch$ (or any TLDP perturbation mechanism) probabilistically allocates $S_i$ to timestamps $t_{i-k+1}$, ..., $t_i$, ..., $t_{i+k-1}, ...$ with probabilities {\it irrespective of} $i$, we can use $\{\mathcal{P}_{1-k}, ..., \mathcal{P}_{-1}, \mathcal{P}_0, \mathcal{P}_1, ..., \mathcal{P}_{k-1}, ...,\mathcal{P}_{n}\}$ to denote these probabilities, where the subscripts are the temporal deviation of the allocated timestamp from the original timestamp. Fig.~\ref{fig:switch_example}(b) illustrates the allocating probabilities of $S_3$ to timestamps $\{t_1, t_2, t_3, t_4, t_5\}$. Note that $n\ge k-1$ and can be as large as infinity, since a value can be repeatedly allocated to its forward timestamps (although with rapidly decreasing probabilities).\footnote{As will be elaborated in Theorem~\ref{theorem:naive_switch}, privacy analysis only needs the first $2k-1$ allocating probabilities, i.e., $\{\mathcal{P}_{1-k}, ..., \mathcal{P}_{-1}, \mathcal{P}_0, \mathcal{P}_1, ..., \mathcal{P}_{k-1}\}$, which spans across two adjacent sliding windows.} In this paper, we collectively call these probabilities the {\bf allocating probability distribution} of $RanSwitch$, the derivation of which lays the foundation of privacy and utility analysis.

%In the above definition, the subscript of each allocating probability indicates the absolute deviation from the original timestamp to the allocated one. The derivation of allocating probability distribution relies on the accumulation effect of perturbation probability distribution throughout an infinite time series, and is the foundation of theoretical privacy and utility analysis.

\subsection{Allocating Probability Distribution in $RanSwitch$}
\label{sec:baseline_probability}
We now derive the allocating probability distribution of $RanSwitch$. For each value $S_i$, $RanSwitch$ can allocate it backward, stationarily, or forward. As such, we derive the three probabilities separately.
\begin{itemize}
	\item {\bf Backward probability} $\{\mathcal{P}_{1-k}, \mathcal{P}_{2-k},..., \mathcal{P}_{-1}\}$, when $S_i$ is dispatched to a previous timestamp within a sliding window. 	
	To start with, $\mathcal{P}_{1-k}$ denotes the probability when $S_{i-k+1}$ switches with $S_i$ so that $S_i$ is released at timestamp $t_{i-k+1}$. Hence $\mathcal{P}_{1-k}=q$. The next probability $\mathcal{P}_{2-k}$ is when $S_i$ is released at $t_{i-k+2}$, which occurs when $S_i$ has not been released at $t_{i-k+1}$ and is then selected at timestamp $t_{i-k+2}$. Hence $\mathcal{P}_{2-k}=q(1-q)$. Similarly, we can derive other backward probabilities as
	$
%		\label{eq:switch_case1}
		\mathcal{P}_j =q(1-q)^{k-1+j},
	$
	where $j\in\{1-k, 2-k, ..., -1\}$.
	
	\item {\bf Current probability} $\mathcal{P}_0$, when the value $S_i$ stays at the current timestamp. This occurs only when $S_i$ has not been released until $t_i$ (w.p. $(1-q)^{k-1}$) and is then released at $t_i$  (w.p. $p$). As such,
	$
%		\label{eq:switch_case2}
		\mathcal{P}_0=p(1-q)^{k-1}.
	$
	
	\item {\bf Forward probability} $\{\mathcal{P}_1, \mathcal{P}_2, ..., \mathcal{P}_{k-1}\}$, when the value $S_i$ is released after $t_i$ within a sliding window. Similar to the derivation of backward probability, for any $j \in \{1, 2, ..., k-1\}$, we have 
	$
%		\label{eq:switch_case3}
		\mathcal{P}_j=q(1-q)^{k-1+j}.
	$
\end{itemize}

%\begin{lemma}
%	\label{label:switch_forward}
%	The forward probability within a sliding window is $\mathcal{P}_j=q(1-q)^{k-1+j}$, where $j \in \{1, 2, ..., k-1\}$.
%\end{lemma}

Combining the above three cases, we reach the following Theorem~\ref{theorem:ranswitch_allocating_pro} on allocating probability distribution in $RanSwitch$ for all the $2k-1$ timestamps.

\begin{theorem}
	\label{theorem:ranswitch_allocating_pro}
	For $1-k \le j \le k-1$, the allocating probability distribution of $RanSwitch$ mechanism is
	\begin{align}
		\label{eq:dispatching_pro}
		\mathcal{P}_j =
		\begin{cases}
			p(1-q)^{k-1},  & \text{if} \ j=0\\
			q(1-q)^{k-1+j},  & \text{otherwise}
		\end{cases}
	\end{align}
\end{theorem}

\subsection{Privacy Analysis of RanSwitch}
\label{sec:baseline_privacy}

%In this subsection, we establish the privacy guarantee of $RanSwitch$ mechanism. 
Based on Theorem~\ref{theorem:ranswitch_allocating_pro}, the following theorem proves that $RanSwitch$ satisfies $(\epsilon, \delta)$-TLDP.

\begin{theorem}
	\label{theorem:naive_switch}
	Given a sliding window of length $k$ and the probabilities $p$ and $q$, the mechanism $RanSwitch$ satisfies $(\epsilon, \delta)$-TLDP, where $\epsilon=\ln \frac{p^2(1-q)^{2(k-1)}-q}{q^2(1-q)^{2(k-1)}}$ and $\delta=q$.
\end{theorem}
\begin{proof}
	For any two neighboring time series $S=\{S_1, S_2, ..., S_i, ..., S_j, ...\}$ and $S^\prime=\{S_1, S_2, ..., S_i^\prime, ..., S_j^\prime, ...\}$, let $t_i$ and $t_j$ denote the two timestamps when $S$ and $S^\prime$ differ, and $|i-j|< k$. As each value may be allocated to a forward/backward timestamp within a sliding window or any forward timestamp beyond this window, the difference of output time series occurs when the value $S_i$ in $S$ is allocated to a forward timestamp which $S_i^\prime$ in $S^\prime$ cannot reach. Thus,
	\begin{align*}
		\delta = \max\{\mathcal{P}_{1-k}, \mathcal{P}_{2-k}, ..., \mathcal{P}_{-2}, \mathcal{P}_{-1}\} = q
	\end{align*}

	To derive $\epsilon$, for any output time series $R$, $\epsilon$ must satisfy
	\begin{align*}
		\epsilon = \mathop{\sup}_{S, S^\prime, R} \ln  \frac{\mathrm{Pr}[\mathcal{A}(S)=R]-\delta}{\mathrm{Pr}[\mathcal{A}(S^\prime)=R]}
	\end{align*}

	For any output $R$ by $RanSwitch$, suppose $S_i$ (or $S^\prime_j$) is allocated to timestamp $t_{\alpha}$, and $S_j$ (or $S^\prime_i$) is allocated to timestamp $t_{\beta}$. Then,
	$
		\frac{\mathrm{Pr}[\mathcal{A}(S)=R]-\delta}{\mathrm{Pr}[\mathcal{A}(S^\prime)=R]}
		= \frac{
			\mathcal{P}_{\alpha-i} \cdot \mathcal{P}_{\beta-j} - \delta
		}{
			\mathcal{P}_{\alpha-j} \cdot \mathcal{P}_{\beta-i}
		}
	$.

	According to Eq.~\ref{eq:dispatching_pro}, except for the case of $j=0$, the allocating probability decreases as $j$ increases. So $\mathcal{P}_{1-k}>\mathcal{P}_{k-1}$. Depending on whether they are larger or smaller than $P_0$, we derive $\epsilon$ separately: 
	\begin{itemize}
		\item When $q\le p(1-q)^{k-1}$, $\mathcal{P}_{k-1}<\mathcal{P}_{1-k}<\mathcal{P}_0$ holds. So
		\begin{align}
			\label{eq:naive_switch_case1}
			\epsilon = \ln \frac{\mathcal{P}_0^2-\delta}{\mathcal{P}_{\alpha-\beta} \cdot \mathcal{P}_{\beta-\alpha}}
			=\frac{p^2(1-q)^{2(k-1)}-q}{q^2(1-q)^{2(k-1)}}
		\end{align}
	
		\item When $q > p(1-q)^{k-1}$, $\mathcal{P}_{1-k}>\mathcal{P}_0$ holds. Since $p>q$, $p(1-q)^{k-1}>q(1-q)^{2(k-1)}$ always holds. Therefore, $\mathcal{P}_{k-1}<\mathcal{P}_0<\mathcal{P}_{1-k}$. So
		\begin{align}
			\label{eq:naive_switch_case2}
%			\frac{\mathrm{Pr}[\mathcal{A}(S)=R]-\delta}{\mathrm{Pr}[\mathcal{A}(S^\prime)=R]}
			\epsilon=\max \{
				\ln \frac{\mathcal{P}_0 \cdot \mathcal{P}_{\beta-j}-\delta}{\mathcal{P}_{i-j} \cdot \mathcal{P}_{\beta-i}},
				\ln \frac{\mathcal{P}_0^2-\delta}{\mathcal{P}_{\alpha-\beta} \cdot \mathcal{P}_{\beta-\alpha}}
			\}
		\end{align}
		As both $\frac{\mathcal{P}_0 \cdot \mathcal{P}_{\beta-j}-\delta}{\mathcal{P}_{i-j} \cdot \mathcal{P}_{\beta-i}} =\frac{p(1-q)^{2(k-1)+\beta-j}-1}{q(1-q)^{2(k-1)+\beta-j}}<0$ and
		$
		\frac{\mathcal{P}_0^2-\delta}{\mathcal{P}_{\alpha-\beta} \cdot \mathcal{P}_{\beta-\alpha}}
		<\frac{p^2(1-q)^{k-1}-p}{q^2(1-q)^{k-1}}
		<0
		$
		always hold when $q > p(1-q)^{k-1}$, any privacy budget $\epsilon>0$ is sufficient to satisfy Eq.~\ref{eq:naive_switch_case2}.
	\end{itemize}
	
Therefore, Eq.~\ref{eq:naive_switch_case1} becomes the only requirement of setting privacy budget, i.e., $\epsilon=\ln \frac{p^2(1-q)^{2(k-1)}-q}{q^2(1-q)^{2(k-1)}}$.
\end{proof}

\subsection{Utility Analysis of RanSwitch}
\label{sec:baseline_cost}
In Sec.~\ref{sec:definition}, there are four costs that collectively determine the utility of a temporal perturbation approach. Fortunately, $RanSwitch$ mechanism only involves misalignment cost, as there are no missing, empty or repeated values. As such, the expectation of the total cost $\mathbb{E}[C]$ is
\begin{align}
	\label{eq:ranswitch_utility_cost}
	\begin{split}
		\quad \mathbb{E}[C]
		= -D \sum\nolimits_{j=1-k}^{-1} j \mathcal{P}_j + D\sum\nolimits_{j=1}^{\infty} j \mathcal{P}_j %+ kD \mathcal{P}_{\infty},
	\end{split}
\end{align}
Note the second term means the count of (forward) misaligned timestamps $j$ spans from $1$ to $\infty$. As the first two sliding window (i.e., $1-k\le j \le k-1$) dominates the total cost, we can derive a lower bound of it as
\begin{align}
	\label{eq:ranswitch_cost}
	\begin{split}
	\mathbb{E}[C]&<\sum\nolimits_{j=1-k}^{-1} (-j)\cdot P_j + \sum\nolimits_{j=1}^{k-1} j\cdot P_j  \\ 
%	&= \sum_{j=1-k}^{-1} (-j)\cdot q(1-q)^{k-1+j} + \sum_{j=1}^{k-1} j\cdot q(1-q)^{k-1+j}  \\
%	&= q\sum_{j=1}^{k-1} (k-j)\cdot (1-q)^{j-1} + q(1-q)^k\sum_{j=1}^{k-1} j\cdot (1-q)^{j-1} \\
%	&= q\sum_{j=1}^{k-1} k \cdot (1-q)^{j-1} - q\sum_{j=1}^{k-1} j \cdot (1-q)^{j-1} + q(1-q)^k\sum_{j=1}^{k-1} j\cdot (1-q)^{j-1} \\
%	&= kq \cdot \frac{1-(1-q)^{k-1}}{1-(1-q)} - q(1-(1-q)^k)\sum_{j=1}^{k-1} j \cdot (1-q)^{j-1}  \\
%	&= k(1-(1-q)^{k-1}) - q(1-(1-q)^k)\sum\nolimits_{j=1}^{k-1} j (1-q)^{j-1} 
%	&= \!k(1\!-\!(1\!-\!q)^{k\!-\!1})D
%- q(1\!-\!(1\!-\!q)^k)\!\sum\nolimits_{j\!=\!1}^{k\!-\!1} j (1\!-\!q)^{j\!-\!1} D \nonumber \\
%	&= (k(1-(1-q)^{k-1}) - \frac{(1-(1-q)^k)^2}{q} + (k-1)(1-q)^{k-1}(1-(1-q)^k))D \\
	&= \left(k\gamma - \frac{\gamma^2}{q} + (k-1)(\gamma-\gamma^2)\right)D
	\end{split}
\end{align}
where $\gamma=1-(1-q)^k$.

From Eq.~\ref{eq:ranswitch_cost}, the lower bound of the total cost is approximately proportional to $k$. As such, if $k$ is large enough, $\mathbb{E}[C]$ becomes very close to $kD$, which means the asymptotic utility of $RanSwitch$ could be poor with respect to large $k$. To illustrate this, in Fig.~\ref{fig:switch_probability} we plot its allocating probability distribution in green line for $k=10$ and the privacy budget $\epsilon=2$.
%For $StaSwitch$ mechanism, we just leave it leave it for the moment and will discuss it in Sec.~\ref{sec:staswithc_cost} for utility analysis.
%The green line in Fig.~\ref{fig:switch_probability} shows the allocating probability distribution of $RanSwitch$ mechanism, of which we observe three drawbacks. First of all, according to Theorem~\ref{theorem:naive_switch}, the relaxation parameter $\delta=P_{-9}$, which is the largest allocating probability among the first $k-1$ ones. This makes $RanSwitch$ can only provide relatively weak privacy guarantee.
We observe that the dispatching probabilities beyond the sliding window, i.e., $\mathcal{P}_{10}$, ... $\mathcal{P}_{19}$, ..., are non-negligible. They form a long tail and dominate $\mathbb{E}[C]$, which is around $4.54$ \footnote{We set the maximum misaligned timestamps to $2k$ in Fig.~\ref{fig:switch_probability}} according to Eq.~\ref{eq:ranswitch_utility_cost}. Furthermore, for those $\mathcal{P}_j$ within the sliding window, the probabilities significantly decrease as $j$ increases, which pushes up the upper bound of the ratio of any two allocating probabilities and thus incurs heavy perturbation to guarantee privacy. %Or equivalently, when given a specific privacy budget, $RanSwitch$ sacrifices its data utility. This can be reflected from

As a comparison, we also plot the allocating probability distribution of the {\it ideal mechanism} in blue line, which splits the probabilities equally among all timestamps except $\mathcal{P}_0$, i.e., $\mathcal{P}_{-9}=\mathcal{P}_{-8}=...=\mathcal{P}_{-1}=P_1=...=\mathcal{P}_8=\mathcal{P}_9=(1-\mathcal{P}_0)/18$. As such, $\mathbb{E}[C]$ is only around $3.82$ according to Eq.~\ref{eq:ranswitch_cost}. Although this ideal mechanism cannot be designed in practice, because the perturbation probability distribution is accumulated throughout an infinite time series, it inspires us to design a close-to-ideal perturbation mechanism in the next section, namely $StaSwitch$, with a more balanced allocating probability distribution than $RanSwitch$.

%Note that this is an optimal setting from the persp ective of privacy guarantee~\cite{wang2019collecting}, because the privacy budget is derived from the maximum ratio of any two allocating probabilities.

\begin{figure}[h]
	\centering{
		\includegraphics[width=0.8\linewidth]{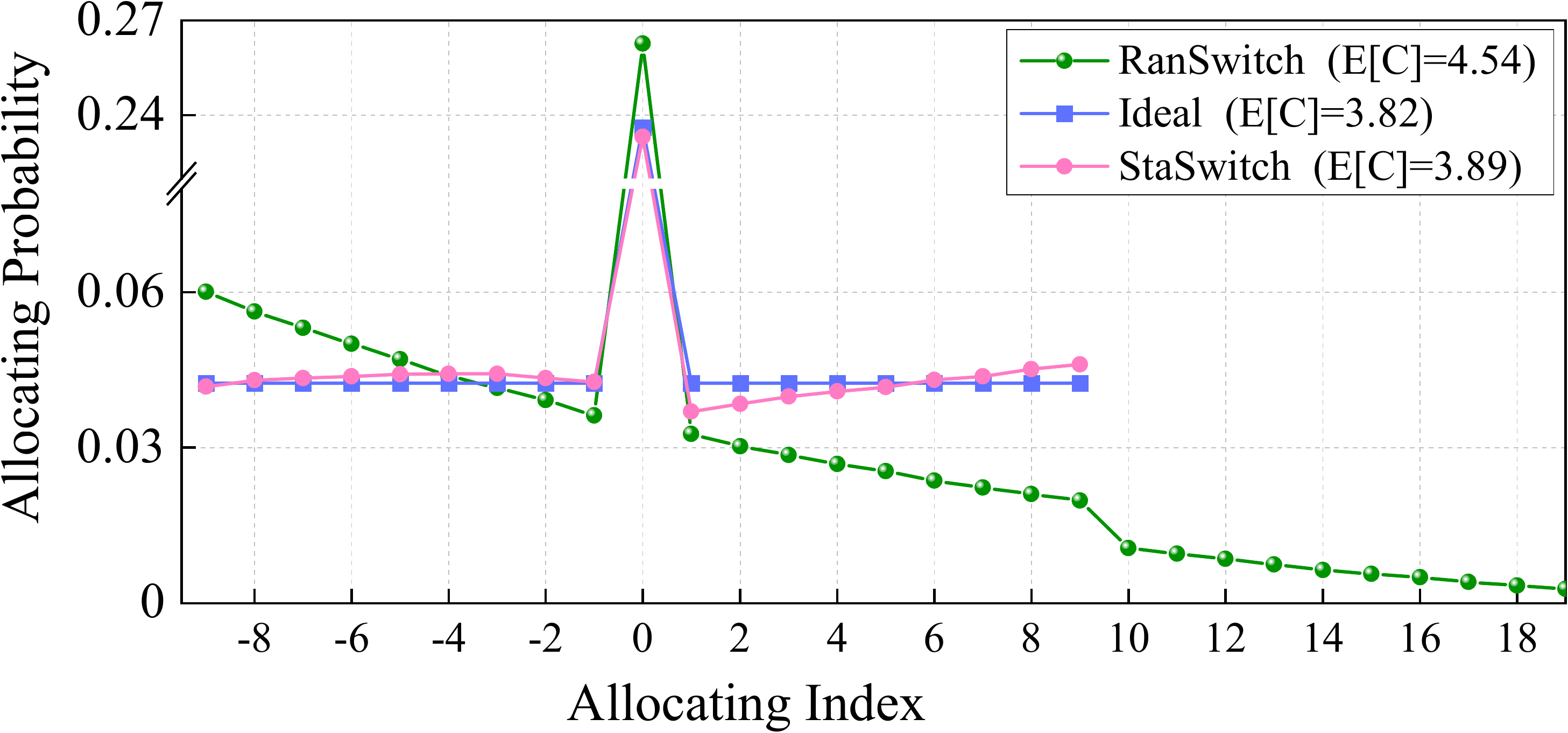}}
	\vspace{-0.1in}
	\caption{Allocating probability distribution ($k=10$ and $\epsilon=2$)}
	\vspace{-0.1in}
	\label{fig:switch_probability}
\end{figure}

%% file: staswitch.tex
The root cause of $RanSwitch$'s high misalignment cost lies in the repeated deferment of a value through multiple {\it switch} operations. For example, a value at timestamp $t_i$ is first forward switched to $t_{i+j}$, and then switched to $t_{i+j+l}$, so on and so forth. To reduce this cost, in this section we propose $StaSwitch$ which keeps track of the switch state of each value and guarantees its final allocated timestamp is still within the initial sliding window. That is, if a value has never been delayed, its allocating space can be the whole sliding window; otherwise, for a value that has already been delayed for $l (l<k)$ timestamps, its new allocation space is limited to the first $k-l$ timestamps of the sliding window.

In what follows, we will first present the perturbation protocol of $StaSwitch$ in Sec.~\ref{sec:staswitch_protocol}, and then derive its allocating probability distribution in Sec.~\ref{sec:staswitch_probability}, followed by privacy analysis in Sec.~\ref{sec:staswithc_privacy} and utility analysis in Sec.~\ref{sec:staswithc_cost}.

\subsection{Perturbation Protocol of StaSwitch}
\label{sec:staswitch_protocol}
Given a sliding window of length $k$, for each value $S_i$, $StaSwitch$ mechanism allows at most $k-1$ timestamps delay. To guarantee this,  at timestamp $t_i$, $StaSwitch$ randomly selects a timestamp $t_j$ in the sliding window $\{t_i, t_{i+1}, ..., t_{i+k-1}\}$ according to the {\bf perturbation probability distribution} in Eq.~\ref{eq:switch_limit}, and then switches $S_i$ and $S_j$.
\begin{align}
	\label{eq:switch_limit}
	\mathrm{Pr}[t_j] = \begin{cases}
		p+bq, \quad &\text{if} \quad j=0 \\
		q, \quad &\text{if} \quad j \in \{1, 2, ..., k-1-b\} \\
		0, \quad &\text{otherwise}
	\end{cases}
\end{align}
Note that $b$ denotes the timestamps $S_i$ has been delayed. If the value $S_i$ has never been delayed, then $b=0$ and Eq.~\ref{eq:switch_limit} degrades to Eq.~\ref{eq:switch_naive} of $RanSwitch$ mechanism. Otherwise,  $S_i$ cannot be allocated to the last $b$ timestamps $\{t_{k-b}, ..., t_{k-1}\}$ in the sliding window of $t_i$. And their perturbation probabilities, which is $bq$, are reclaimed by $StaSwitch$ and allocated to $t_i$. As such, $StaSwitch$ guarantees each value can be delayed up to $k-1$ timestamps. $StaSwitch$ gets its name from {\bf stateful switch}, as the perturbation probability distribution depends on the current state of delayed timestamps $b$. Table~\ref{table:probability_dis} shows an example under different delayed timestamps $b$ when $k=4$, where each cell shows the perturbation probability of $S_i$ being allocated to $t_{i+j}$.

\begin{table}
	\caption{Perturbation probability distribution when $k=4$}
	\vspace{-0.08in}
	\label{table:probability_dis}
	\centering
	\begin{tabular}{|c|c|c|c|c|}
		\hline
		& \bm{$j=0$}  & \bm{$j=1$} & \bm{$j=2$} & \bm{$j=3$} \\ \hline	
		\bm{$b=0$} & $p$ & $q$ & $q$ & $q$  \\ \hline
		\bm{$b=1$} & $p+q$ & $q$ & $q$ & $0$  \\ \hline
		\bm{$b=2$} & $p+2q$ & $q$ & $0$ & $0$  \\ \hline
		\bm{$b=3$} & $p+3q$ & $0$ & $0$ & $0$  \\ \hline
	\end{tabular}
	\vspace{-0.15in}
\end{table}

\vspace{-0.05in}
\begin{algorithm}
	%\footnotesize
	\small
	\caption{ Perturbation protocol of $StaSwitch$ }
	\begin{tabular}{ll}
		{\bf Input:} & Original time series $S = \{S_1, S_2, ..., S_n, ...\}$\\
		& Sliding window length $k$ \\
		& Perturbation probabilities $p$ and $q $ \\
		{\bf Output:} & Released time series $R = \{R_1, R_2, ..., R_n, ...\}$ \\
%		{\bf Procedure:} &
	\end{tabular}
	\label{alg:staswitch}
	\begin{algorithmic}[1]
		\STATE Initialize a released time series $R=\emptyset$
		\STATE Initialize a vector of delayed timestamps $b=\{0\}^{|S|}$
		\FOR{each timestamp $t_i$ ($i \in \{1, 2, ..., n, ...\}$)}
			\STATE Randomly select an index $j$ from $\mathbb{X} = \{l|i \le l \le i+k-1-b_i\}$ according to Eq.~\ref{eq:switch_limit}
			\STATE Switch $S_i$ and $S_j$
			\STATE $b_j=b_j+j-i$
			\STATE $R_i=S_i$ and release $R_i$
		\ENDFOR
		%\RETURN $R = \{R_1, R_2, ..., R_n, ...\}$
	\end{algorithmic}
\end{algorithm}
\vspace{-0.05in}

Algorithm~\ref{alg:staswitch} shows the pseudo-code of $StaSwitch$ mechanism. The procedure is similar to Algorithm~\ref{alg:ranswitch}, except for vector $b$, which record the current delayed timestamps of each value. At each timestamp $t_i$,  an index $j$ is randomly drawn from Eq.~\ref{eq:switch_limit} (Line 4), and then the current value $S_i$ and the selected one $S_j$ are switched (Line 5). As such, $S_i$ is delayed by $j-i$ timestamps, so $b_j$, the delayed timestamps of $S_j$,  is incremented by $j-i$ (Line 6). Finally, the current value $S_i$ (i.e., the original $S_j$) is released (Line 7).

\subsection{Allocating Probability Distribution in $StaSwitch$}
\label{sec:staswitch_probability}
%Note that under $StaSwitch$, each value's allocated timestamp may influence the state of the subsequent $k-1$ timestamps. For example, if $S_i$ selects timestamp $t_{i+1}$, the delay budget at $t_{i+1}$ becomes $b=1$, changing its the probability distribution across a sliding window to ${\{p+q, q, ..., q, 0\}}$.
We now derive the allocating probability distribution $\mathcal{P}=\{\mathcal{P}_{1-k}, ..., \mathcal{P}_{-1}, \mathcal{P}_0, \mathcal{P}_1, ..., \mathcal{P}_{k-1}\}$ of $StaSwitch$ for privacy and utility analysis. To start with, we first derive the {\it expected} perturbation probability distribution over all $b$'s from Eq.~\ref{eq:switch_limit} as below. %Formally for any $j \in \{0, 1, ..., k-1\}$, %Specifically, we use the following Eq.~\ref{eq:limit_pro} to denote allocating probability distribution of each value in a sliding window $\{S_i, S_{i+1}, ..., S_{i+k-1}\}$ at a stable state.
\begin{align}
	\label{eq:limit_pro}
	\mathbb{E}[\mathrm{Pr}[t_j]] \! = \! \begin{cases}
		p_0=\! \sum_{i=0}^{k-1}P_b^i (p\!+\!iq), &\text{if} \ j=0 \\
		q_j=\! \sum_{i=0}^{k\!-\!j\!-\!1}P_b^i q,  &\text{if} \ 1 \le j \le k\!-\!1
	\end{cases}
\end{align}
where $P_b^i=\mathrm{Pr}[b=i]$. Note that, at timestamp $t_{i+j}$, $P_b^i$ is the probability that the current value comes from $S_j$ (i.e., with $i$ timestamps delay) before switching value $S_{i+j}$. Since there are $i$ such cases, we sum them up as below.
\begin{align}
	\label{eq:empty_pro}
	P_b^i
%		&= P_b^0 \cdot q \prod_{l=1}^{j-1}(1-q_l)
%		+ P_b^1 \cdot q \prod_{l=1}^{j-2}(1-q_l)
%		+ ...
%		+ P_b^{j-2} \cdot q (1-q_1)
%		+ P_b^{j-1} \cdot q \\
	&= q \sum\nolimits_{j=0}^{i-1} P_b^j \cdot \prod\nolimits_{l=1}^{i-1-j}(1-q_l),
\end{align}
where $P_b^0 = \prod_{j=1}^{k-1} (1-q_j)$ denotes the probability that the current value has never been delayed.

By substituting Eq.~\ref{eq:empty_pro} for $P_b^i$  in Eq.~\ref{eq:limit_pro}, we obtain the expected perturbation probability distribution over $p_0$ and all $q_j$. Then we derive the allocating probability distribution as follows. First, we know $\mathcal{P}_{1-k}$ means $S_{i-k+1}$ selects timestamp $t_i$, so $S_i$ is released at $t_{i-k+1}$. Hence $\mathcal{P}_{1-k}=q_{k-1}$. Similarly, we can derive other allocating probabilities as follows. For any $1\le j\le k-1$,
\begin{align}
	\label{eq:staswitch_dispatching_p}
	&\mathcal{P}_{j-k} = q_{k-j}\prod\nolimits_{i=1}^{j-1}(1-q_{k-i}) \nonumber  \\
	&\mathcal{P}_{0} = p \prod\nolimits_{j=1}^{k-1}(1-q_{k-j})	 \\
%		&\mathcal{P}_j = q \sum_{l=0}^{j-1} P_b^l  \left( \Phi(j,l)+ \Psi(j,l) \right)	\\
%		&\mathcal{P}_j = q \sum_{l=0}^{j-1} P_b^l  \left(
%		 	(p+jq)  \prod_{i=1}^{j-l-1}(1-q_i)
%		 	+ \sum_{i=1}^{k-j-1} q_i  \prod_{r=1}^{j-l-1}  (1-q_{r+i})
%		 \right)	\\
	&\mathcal{P}_j = q\! \sum_{l=0}^{j-1}\! P_b^l \! \left(\!
		(p\!+\!jq)  \prod_{i=1}^{j\!-\!l\!-\!1}(1\!-\!q_i)
		\!+\! \sum_{i=1}^{k\!-\!j\!-\!1}\! q_i\!  \prod_{r=1}^{j\!-\!l\!-\!1} \! (1\!-\!q_{r+i}) \!
	\right) \nonumber
\end{align}
%where $\Phi(j,l)=(p+jq)  \prod_{i=1}^{j-l-1}(1-q_i)$ and  $\Psi(j,l)= \sum_{i=1}^{k-j-1} q_i  \prod_{r=1}^{j-l-1}  (1-q_{r+i})$.

To interpret the last equation, $\mathcal{P}_j$, the allocating probability of a value being deviated by $j$ timestamps, is comprised of $j$ joint probabilities, each first allocating a value with $l$ ($0\le l \le j-1$) timestamps deviation (i.e., $P_b^l$ ) and then allocating the same value with another $j-l$ timestamps deviation. The latter probability is further comprised of two terms. The first term, $(p+jq)\prod_{i=1}^{j-l-1}(1-q_i)$, denotes the probability that a value with $l$ timestamps deviation is directly switched to timestamp $t_j$, and the second term, $\sum_{i=1}^{k-j-1} q_i \prod_{r=1}^{j-l-1}  (1-q_{r+i})$, is the probability that the value first selects a timestamp after $t_j$ and then switches it with value $S_j$.

\subsection{Privacy Analysis}
\label{sec:staswithc_privacy}
This subsection establishes the privacy guarantee of $StaSwitch$ mechanism. Lemmas~\ref{lemma:empty_p_increasing} and \ref{lemma:staswitch_allocating_pro} first show the monotonicity of the probability distribution $P_b^i (1 \le i \le k-1)$ and the allocating probability distribution $\mathcal{P}_i (1-k\le i \le k-1)$ respectively, based on which we prove $StaSwitch$ satisfies ($\epsilon$, $\delta$)-TLDP in Theorem~\ref{theorem:staswitch_privacy}.

\begin{lemma}
	\label{lemma:empty_p_increasing}
	For $j\in\{1, 2, ..., k-1\}$, the probability $P_b^j$ monotonously increases with $j$, i.e., $P_b^1<P_b^2<...<P_b^{k-1}$.
\end{lemma}
%\begin{proof}
%	For $j\in\{1, 2, ..., k-2\}$,
%	\begin{align*}
%		&\quad P_b^{j+1} - P_b^j  \\
%		& =  q \sum_{i=0}^{j} P_b^i  \prod_{l=1}^{j-i}(1-q_l)
%			- q \sum_{i=0}^{j-1} P_b^i  \prod_{l=1}^{j-1-i}(1-q_l) \\
%%		&= q P_b^j + q \sum_{i=0}^{j-1} P_b^i \cdot \prod_{l=1}^{j-i}(1-q_l)
%%			- q \sum_{i=0}^{j-1} P_b^i \cdot \prod_{l=1}^{j-1-i}(1-q_l) \\
%		&= q P_b^j + q \sum_{i=0}^{j-1} P_b^i \left(
%				\prod_{l=1}^{j-i}(1-q_l)
%				- \prod_{l=1}^{j-1-i}(1-q_l)
%			\right) \\
%%		&= q P_b^j + q \sum_{i=0}^{j-1} P_b^i
%%				\prod_{l=1}^{j-1-i}(1-q_l) \cdot (1-q_{j-i}-1)   \\
%%		&= q P_b^j - q \sum_{i=0}^{j-1} P_b^i
%%				\prod_{l=1}^{j-1-i}(1-q_l) \cdot q_{j-i}  \\
%		&= q^2 \sum_{i=0}^{j-1} P_b^i  \prod_{l=1}^{j-1-i}(1-q_l)
%			- q \sum_{i=0}^{j-1} P_b^i  \prod_{l=1}^{j-1-i}(1-q_l) q_{j-i}  \\
%		&=q \sum_{i=0}^{j-1} P_b^i (q-q_{j-i}) \prod_{l=1}^{j-1-i}(1-q_l)
%	\end{align*}
%	
%	According to Eq.~\ref{eq:limit_pro}, $q>q_j$ for each $j\in\{1, 2, ..., k-1\}$. Thus, $\quad P_b^{j+1} - P_b^j>0$ always holds, which means $P_b^j$ monotonously increases with $j\in\{1, 2, ..., k-1\}$.
%\end{proof}

\begin{lemma}
	\label{lemma:staswitch_allocating_pro}
	The minimum of $StaSwitch$'s allocating probabilities $ \{\mathcal{P}_{1-k}, ..., \mathcal{P}_{-1}, \mathcal{P}_{0}, \mathcal{P}_{1}, ..., \mathcal{P}_{k-1}\}$ is $\mathcal{P}_{min} = \min\{\mathcal{P}_{1-k}, \mathcal{P}_{-1}, \mathcal{P}_1\}$.
\end{lemma}

\begin{theorem}
	\label{theorem:staswitch_privacy}
	The $StaSwitch$ satisfies ($\epsilon$, $\delta$)-TLDP, where	
	$
	\epsilon = \ln \frac{\frac{p^2}{\sigma} - (p^2-p+2)}{q \left(1+q- \frac{k(1-p)q}{2(1+q)}-\frac{q}{2-p}\right)}
	$, 	
	$\delta = \max\{\mathcal{P}_{1-k}, \mathcal{P}_{2-k}, ..., \mathcal{P}_{-1}\}$,
	$
	\sigma=\frac{(1-p)(1+p+q)(2-p)-q}{2(k-2)(1+q)(2-p)} + \frac{(k-3)q^2 (1-q)^{k-1}}{2}
	$, and $p$ and $q$ are parameters defined in Eq.~\ref{eq:switch_limit}.
\end{theorem}
\begin{proof}	
	Let $t_i$ and $t_j$ denote the two timestamps that neighboring time series $S$ and $S^\prime$ differ. That is, $S_i$ = $S^\prime_j$, and $S_j$ = $S^\prime_i$. In the output $R$, let $t_{\alpha}$ and $t_{\beta}$ denote the timestamps to which the values of $S_i$ and $S_j$ are allocated, respectively. Therefore,
	\begin{align}
		\label{eq:staswitch_privacy_budget}
		\begin{split}		
			\epsilon\! &= \! \mathop{\sup}_{S, S^\prime, R} \! \ln \!
			\frac{\mathrm{Pr}[\mathcal{A}(S)\!=\!R]\!-\!\delta}{\mathrm{Pr}[\mathcal{A}(S^\prime)\!=\!R]}
			= \mathop{\sup}_{\alpha, \beta, i, j} \! \ln \! \frac{
					\mathcal{P}_{\alpha-i} \mathcal{P}_{\beta-j} \!-\! \delta
				}{
					\mathcal{P}_{\alpha-j} \mathcal{P}_{\beta-i}
				}
        \end{split}
    \end{align}

	According to Lemma~\ref{lemma:staswitch_allocating_pro}, the minimum allocating probability is $\mathcal{P}_{min} = \min\{\mathcal{P}_{1-k}, \mathcal{P}_{-1}, \mathcal{P}_1\}$, where $\mathcal{P}_{1-k} = q_{k-1}$, $\mathcal{P}_{-1} = q_1 \prod_{j=2}^{k-1}(1-q_j)$, and $\mathcal{P}_1= \prod_{j=1}^{k-1} (1-q_j) q (p+q + 1-p_0-q_{k-1})$.

Therefore, Eq.~\ref{eq:staswitch_privacy_budget} is reduced to
    \begin{align}
		\label{eq:ratio_limit_bound}
		\begin{split}		
			\epsilon &= \ln \frac{\mathcal{P}_0^2-\delta}{\mathcal{P}_{1} \cdot \mathcal{P}_{-1}}
			\le \ln \left(
					\frac{\mathcal{P}_0^2}{\mathcal{P}_{1} \cdot \mathcal{P}_{-1}} - \frac{1}{\mathcal{P}_1}
				\right) \\
%			&< \ln \frac{p^2 (1-q_1) - q_1(2-p)}{q q_1 (p+q + 1-p_0-q_{k-1})} \\
			&< \ln \frac{p^2 - q_1(p^2-p+2)}{q q_1 (p+q + 1-p_0-\frac{q}{2-p})}
		\end{split}
	\end{align}
	
	From the above inequality, to derive an upper bound of $\epsilon$, we need to have an upper bound of $p_0$ and a lower bound of $q_1$ respectively. First we know that
	\begin{align*}
	P_b^{k\!-\!1} &= q \sum\nolimits_{i=0}^{k-2} P_b^i \prod\nolimits_{l=1}^{k\!-\!2\!-\!i}(1\!-\!q_l)
	< q \! \sum_{i=0}^{k-2} P_b^i
	= q \! - \! q P_b^{k\!-\!1}
	\end{align*}
	Therefore, $P_b^{k-1}<\frac{q}{1+q}$, and hence,
	\begin{align*}
		p_0 &= \sum\nolimits_{i=0}^{k-1}P_b^i  \cdot (p+iq)
			= p + q\sum\nolimits_{i=1}^{k-1}P_b^i  \cdot i  \\
		&< p + q P_b^{k-1} \frac{k(k-1)}{2}
		< p + \frac{k(1-p)q}{2(1+q)}
	\end{align*}
	
	Then we derive a lower bound of $q_1$. For any $j \in \{1, 2, ..., k-2\}$, according to Eq.~\ref{eq:limit_pro} and Lemma~\ref{lemma:empty_p_increasing},
	\begin{align*}
		q_j - q_{j+1} &= \sum\nolimits_{i=0}^{k\!-\!j\!-\!1}P_b^i q
			- \sum\nolimits_{i=0}^{k\!-\!j\!-\!2}P_b^i q
		= q P_b^{k\!-\!j\!-\!1} > q P_b^1
	\end{align*}	
	Thus,
	$
		q_1-(k\!-\!2)q P_b^1 > q_2-(k\!-\!3)q P_b^1
		> \cdots
		> q_{k-2}-q P_b^1 =q_{k-1}
	$.
	The first term $q_1-(k\!-\!2)P_b^1$ must be greater than the average of the first $k-2$ terms, i.e.,
	\begin{align*}
	q_1-(k-2)q P_b^1
	&> \frac{\sum_{i=1}^{k-2}q_i -q P_b^1 \sum_{i=1}^{k-2} i}{k-2} \\
%enen	&= \frac{1-p_0-q_{k-1} - \frac{(k-1)(k-2)}{2}qP_b^1}{k-2} \\
	&= \frac{1-p_0-q_{k-1}}{k-2} - \frac{k-1}{2}q P_b^1
	\end{align*}	
	Therefore,
	\begin{align*}
		q_1 &> \frac{1-p_0-q_{k-1}}{k-2} - \frac{k-1}{2}q P_b^1 + (k-2)qP_b^1 \\
	%	&= \frac{1-p_0-q_{k-1}}{k-2} + \frac{(k-3)q P_b^1}{2} \\
	%	&= \frac{1-p_0-q_{k-1}}{k-2} + \frac{(k-3)q^2 P_b^0}{2} \\
	%	&= \frac{1-p_0-q_{k-1}}{k-2} + \frac{(k-3)q^2 \prod_{j=1}^{k-1}(1-q_j)}{2} \\
	%	&> \frac{1-p_0-q_{k-1}}{k-2} + \frac{(k-3)q^2 (1-q)^{k-1}}{2}  \\
	%	&> \frac{1-p_0-\frac{q}{2-p}}{k-2} + \frac{(k-3)q^2 (1-q)^{k-1}}{2}  \\
	%	&= \frac{(1-p_0)(2-p)-q}{(k-2)(2-p)} + \frac{(k-3)q^2 (1-q)^{k-1}}{2} \\
		&> \frac{(1\!-\!p)(1\!+\!p\!+\!q)(2\!-\!p)\!-\!q}{2(k-2)(1+q)(2-p)}
			+ \frac{(k\!-\!3)q^2 (1\!-\!q)^{k-1}}{2} = \sigma
	\end{align*}

	Then following Eq.~\ref{eq:ratio_limit_bound}, the upper bound of $\epsilon$ becomes $\ln \frac{\frac{p^2}{\sigma} - (p^2-p+2)}{q \left(1+q- \frac{k(1-p)q}{2(1+q)}-\frac{q}{2-p}\right)}$.
%	\begin{align*}
%		\epsilon
%%		&< \ln \frac{p^2 - q_1(p^2-p+2)}{
%%			q q_1 \left(p+q + 1-\left(p + \frac{k(1-p)q}{2(1+q)}\right)-\frac{q}{2-p}\right)
%%		}  \\
%		&< \ln \frac{\frac{p^2}{\sigma} - (p^2-p+2)}{q \left(1+q- \frac{k(1-p)q}{2(1+q)}-\frac{q}{2-p}\right)}
%	\end{align*}
	The proof for $\delta$ follows that of Theorem~\ref{theorem:naive_switch}, where $\delta$ must cover the difference between the output space of any two neighboring time series $S$ and $S^\prime$, that is,
$
		\delta = \max\{\mathcal{P}_{1-k}, \mathcal{P}_{2-k}, ..., \mathcal{P}_{-2}, \mathcal{P}_{-1}\}
		%= \mathcal{P}_m
$.
	%where $m=\lceil \frac{\sqrt{q^2-4p+8}+3q-2}{2q}-0.5 \rceil-k$.
\end{proof}

Theorem~\ref{theorem:staswitch_privacy} provides a close-form of $\epsilon$, but not $\delta$. The following corollary gives a good estimation on it.

\begin{corollary}
	\label{lemma:staswitch_pro_max}
	The maximum of the first $k-1$ allocating probabilities can be approximated by $\mathcal{P}_{max}$, where $max=\lceil \frac{\sqrt{q^2-4p+8}-2}{2q} \rceil-(k+1)$.
\end{corollary}
\begin{proof}
	According to Eq.~\ref{eq:empty_pro}, $P_b^1 = \prod_{j=1}^{k-1} (1-p_j) q = q P_b^0$. For $j\in \{2, 3, ..., k-1\}$, by approximating $P_b^j$ by $P_b^1$ , i.e.,  $P_b^j\approx P_b^1$, we have
	$
	\sum_{j=0}^{k-1}P_b^i
	\approx P_b^0+(k-1)qP_b^0
	= (2-p)P_b^0
	$.
	
	By solving $(2-p)P_b^0\approx 1$, we know for $j\in \{1, ..., k-1\}$,
	\begin{align}
		\label{eq:p0pj_approx}
		P_b^0 \approx \frac{1}{2-p}, \
		P_b^j \approx \frac{q}{2-p}
	\end{align}
	By substituting Eq.~\ref{eq:p0pj_approx} in Eq.~\ref{eq:limit_pro}, we have
	\begin{align}
		\label{eq:qj_approx}
		q_j \approx \frac{q+(k-j-1)q^2}{2-p}
	\end{align}
	
	Then according to Eq.~\ref{eq:staswitch_dispatching_p},
	\begin{align*}
		\mathcal{P}_{j-k}\!-\!\mathcal{P}_{j-k-1}\!=\!0 
		&\Leftrightarrow 
		q_{k-j} - q_{k-j+1} - q_{k-j}\cdot q_{k-j+1} \! = \! 0	\\
		&\Leftrightarrow
		\frac{(1+(j-1)q)(1+(j-2)q)}{2-p} = 1 \\
		&\Leftrightarrow 
		j=\frac{\sqrt{q^2-4p+8}+3q-2}{2q}
	\end{align*}
	Therefore, the maximum probability is $\mathcal{P}_{max}$, where $max=\lceil j-0.5 \rceil-k=\lceil \frac{\sqrt{q^2-4p+8}-2}{2q}\rceil-(k+1)$.
\end{proof}

\subsection{Utility Analysis}
\label{sec:staswithc_cost}
The only cost of $StaSwitch$ is misalignment cost. So the expectation of the total cost is 
\begin{align}
	\label{eq:staswitch_cost}
	&\quad \mathbb{E}[C]
	=  D \left(\sum\nolimits_{j=1}^{k-1} (-j) \mathcal{P}_{j-k} + \sum\nolimits_{j=1}^{k-1} j \mathcal{P}_j \right)
\end{align}
To derive $\mathbb{E}[C]$, we substitute Eqs.~\ref{eq:p0pj_approx} and \ref{eq:qj_approx} in Eq.~\ref{eq:staswitch_dispatching_p}, then
\begin{align}
		\label{eq:allocating_p_approx}
		&\mathcal{P}_{j-k} \approx \frac{q+(j-1)q^2}{2-p} \prod\nolimits_{i=1}^{j-1}(1-\frac{q+(i-1)q^2}{2-p})  \nonumber  \\
		&\mathcal{P}_{0} \approx p \prod\nolimits_{j=1}^{k-1}(1-\frac{q+(j-1)q^2}{2-p})	 \\
		&\mathcal{P}_j \approx   \frac{q\left( \Phi(j,0) + \Psi(j,0) \right)}{2-p}
		+   \frac{q^2 \sum_{l=1}^{j-1} \left(  \Phi(j,l) + \Psi(j,l) \right)}{2-p},	 \nonumber
\end{align}
where
\begin{align*}
	&\Phi(j,l)=(p+jq)  \prod\nolimits_{i=1}^{j-l-1}(1-\frac{q+(k-i-1)q^2}{2-p}) \\
	&\Psi(j,l)= \sum_{i=1}^{k\!-\!j\!-\!1} \frac{q\!+\!(k\!-\!i\!-\!1)q^2}{2-p} \prod_{r=1}^{j\!-\!l\!-\!1}  (1\!-\!\frac{q+(k\!-\!r\!-\!i\!-\!1)q^2}{2-p})
\end{align*}

By substituting Eq.~\ref{eq:allocating_p_approx} in Eq.~\ref{eq:staswitch_cost}, we obtain an approximation of the expected total cost $\mathbb{E}[C]$ of $StaSwitch$ mechanism. For example, we can obtain $\mathbb{E}[C]=3.89$ in Fig.~\ref{fig:switch_probability}. Compared with $\mathbb{E}[C]=3.82$ for the ideal mechanism, $StaSwitch$ achieves very similar utility. %This attributes to the stateful switch operation exploited in $StaSwitch$. First, any value will not be allocated beyond a sliding window, which reduces alignment cost significantly. On the other hand, $StaSwitch$ is capable of maintaining a balanced allocating probability distribution, which ensures the upper bound of any two allocating probabilities is not too large. These two aspects helps to enhance the utility of the released time series significantly.

%% file: experiment.tex
In this section, we evaluate and compare $RanSwtich$ and $StaSwitch$ with state-of-the-art TLDP mechanism, namely (Extended) Threshold mechanism (TM/ETM)~\cite{ye2021beyond}, and value-perturbation mechanisms for time series such as Randomized Response~\cite{warner1965randomized} and Piecewise mechanism~\cite{wang2019collecting}.
\subsection{Experimental Setting}

%\subsubsection{Datasets}
{\bf Datasets.} We conduct experiments on two real and one synthetic time series datasets.
\begin{itemize}
	\item {\it US stock}~\cite{USstock} consists of historical daily prices of 14,058 trading days. We first extract all daily close price as a numerical time series {\it Stock-N}, and then derive another binary time series {\it Stock-B} by comparing each close price with its previous day, so each value indicates ``up'' or ``down'' of daily stock price.
	
	\item {\it Trajectory}~\cite{Taxitrajectory} consists of 6,307 taxi trajectories, each of which has GPS coordinates in a 15-second interval and has at least 300 timestamps.
	
	\item {\it SyntheticTS} is a generated synthetic time series that consists of $10^6$ timestamps, whose values are integers randomly drawn from $[0,100]$.
\end{itemize}

%\subsubsection{Experiment Design}
{\bf Experiment Design.} We design two sets of experiments. The first set evaluates the overall cost of the three TLDP mechanisms under various datasets and parameters, including the sliding window length $k$ and privacy budget $\epsilon$. The second set compares their utility in three real-world applications of popular time series manipulation, namely, simple moving average, frequency counting, and trajectory clustering.

We implement all mechanisms in Java and conduct experiments on a desktop computer with Intel Core i9-9900K 3.60 GHz CPU, 64G RAM running Windows 10 operating system. %All measurements are averaged from 1000 trials.
%To measure the accuracy of ECG heartbeat categorization, we adopt the architecture in~\cite{kachuee2018ecg}, which is based on convolutional neural network (CNN) and use the same parameter settings.
%To measure the accuracy of trajectory clustering, we use the classic $k$-medoids algorithm~\cite{van2003new, park2009simple} to divide all trajectories into $6$ clusters.

\begin{figure*}[t]
	\centering
	\begin{minipage}{0.24\linewidth}
		\centerline{\includegraphics[width=\linewidth]{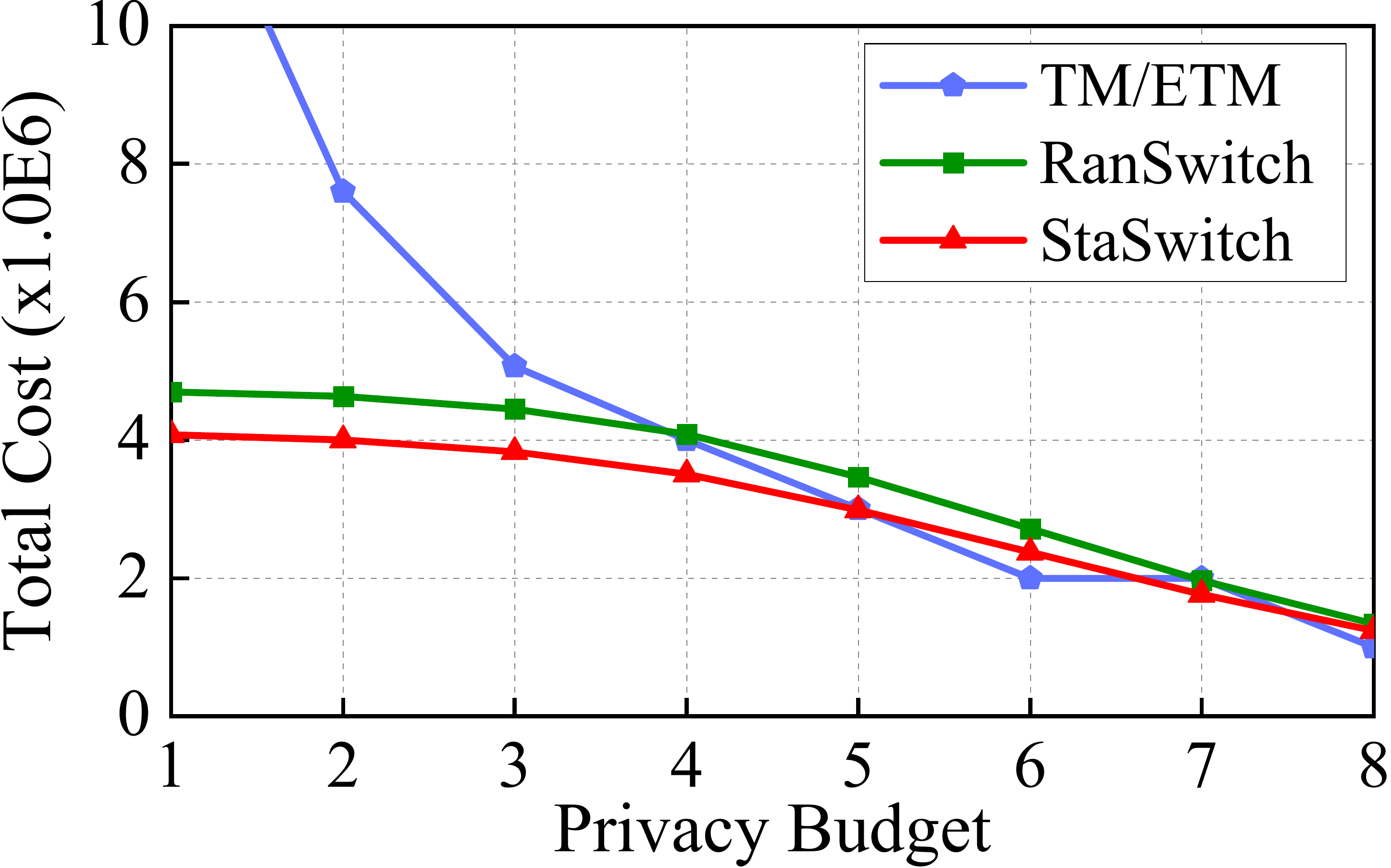}}
		\centerline{\footnotesize{(a) Total cost of $k=10$}}
	\end{minipage}
	%\hspace{-0.1in}
	\begin{minipage}{0.24\linewidth}
		\centerline{\includegraphics[width=\linewidth]{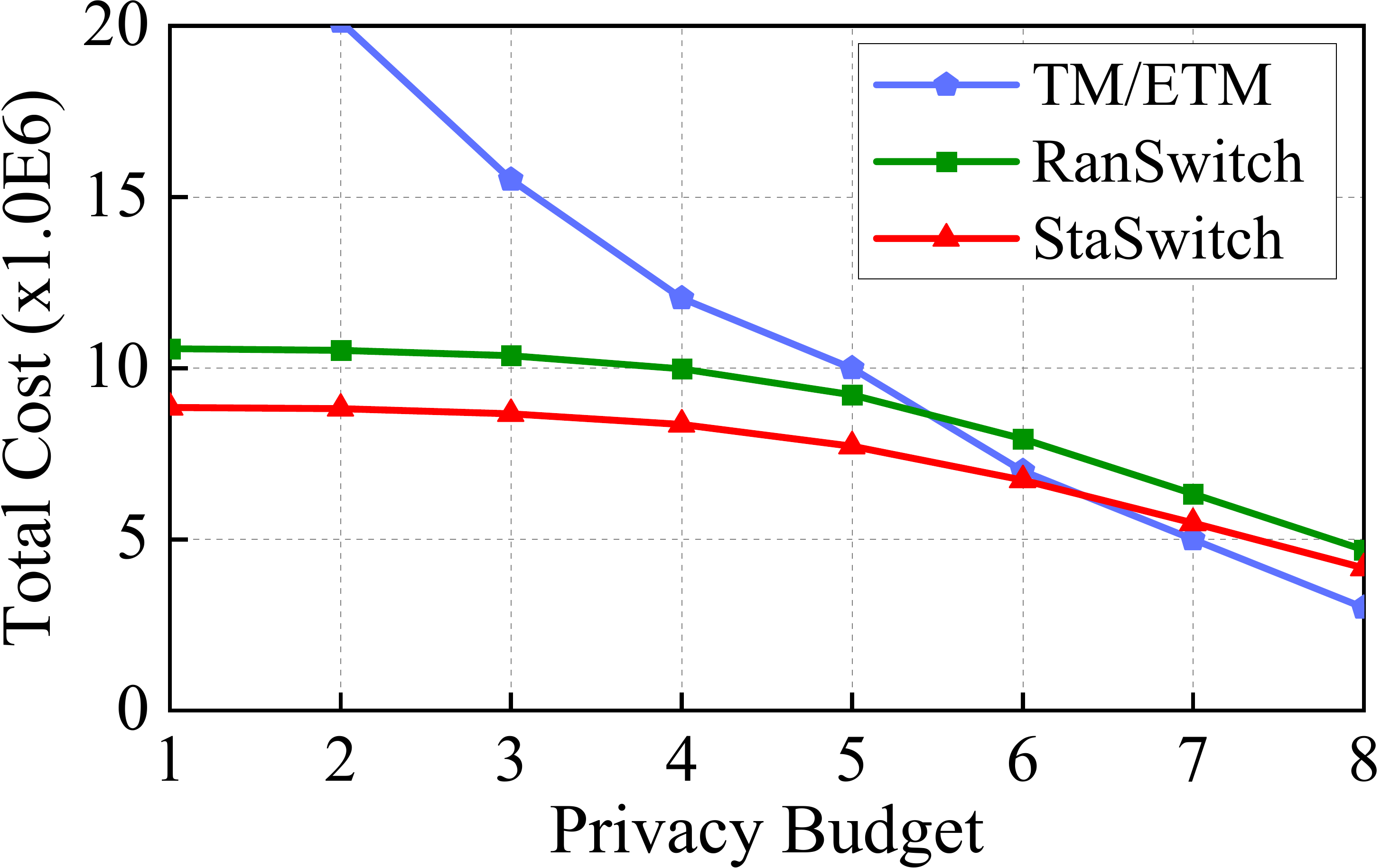}}
		\centerline{\footnotesize{(b) Total cost of $k=20$}}
	\end{minipage}
	%\hspace{0.1in}
	\begin{minipage}{0.24\linewidth}
		\centerline{\includegraphics[width=\linewidth]{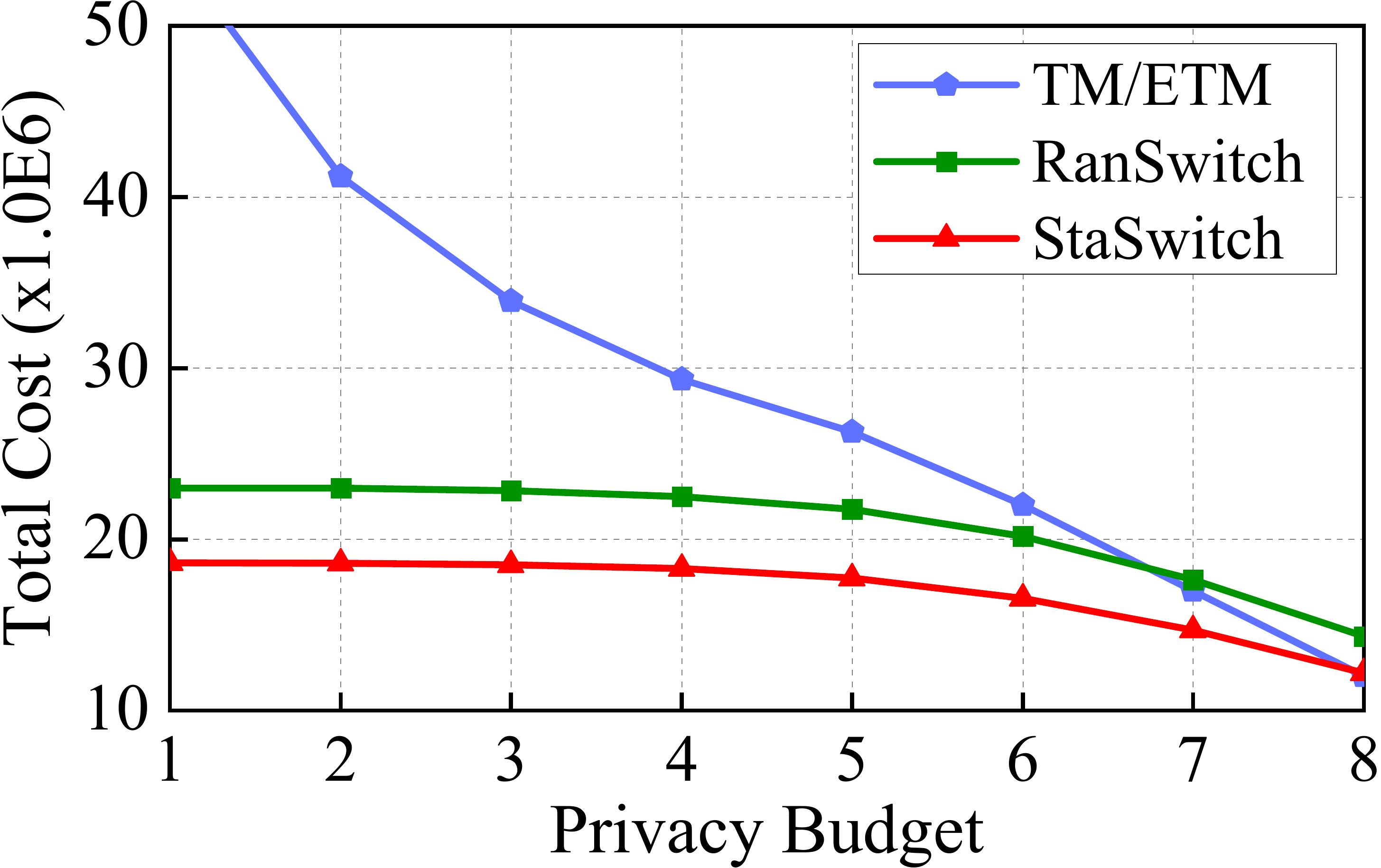}}
		\centerline{\footnotesize{(c) Total cost of $k=40$}}
	\end{minipage}
	\begin{minipage}{0.24\linewidth}
		\centerline{\includegraphics[width=\linewidth]{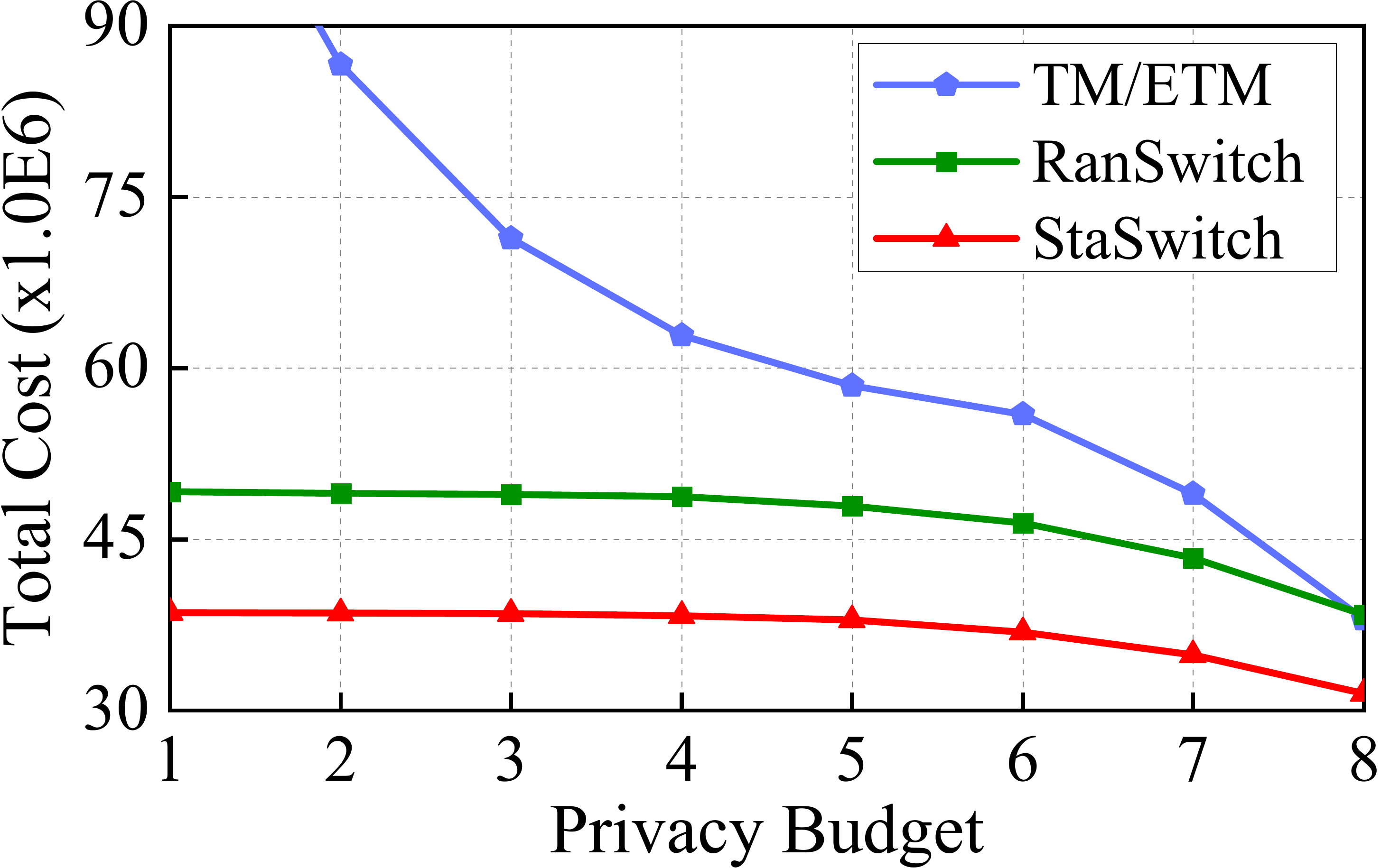}}
		\centerline{\footnotesize{(d) Total cost of $k=80$}}
	\end{minipage}
	\vspace{-0.02in}
	\caption{Overall comparison of total cost of different TLDP mechanisms on {\it SyntheticTS}}
	\vspace{-0.1in}
	\label{fig:total_cost}
\end{figure*}

\subsection{Overall Cost Evaluation}

This subsection evaluates the total cost of three TLDP mechanisms, i.e., $RanSwitch$, $StaSwitch$ and TM/ETM. According to Sec.~\ref{sec:definition}, the total cost is
\[
C = D\sum\nolimits_{i} l_i + M\cdot n_1 + N \cdot n_2 + E \cdot n_3,
\]
where $l_i$ denotes each value $S_i$'s count of timestamps deviated after perturbation, and $n_1$, $n_2$ and $n_3$ are the  numbers of missing, empty and repeated values, respectively. In the experiment, we set the unit cost of misalignment $D=1$, and set unit cost of missing, repetition and empty $M=N=E=k$, which means these costs are as worse as allocating a value to the endpoint of the sliding window.

Fig.~\ref{fig:total_cost} plots the total cost of three mechanisms on the dataset {\it SyntheticTS}, by varying the length of sliding window from $10$ to $80$, and the privacy budget from $1$ to $8$.%\footnote{This range of privacy budget serves to verify our analysis of utility cost. Later in real applications, as a common practice, we focus on the privacy budget from $1$ to $8$. }.
\footnote{The privacy parameter $\delta$ can be derived by Theorems~\ref{theorem:naive_switch} (for $RanSwitch$) and \ref{theorem:staswitch_privacy} (for $StaSwitch$) according to $\epsilon$ and $k$, so it is not shown in the figure.}
Overall, $StaSwitch$ performs the best and consistently outperforms $RanSwitch$ in all cases, thanks to its stateful switch operation. The gain of $StaSwitch$ becomes more eminent with the increasing $k$, and exceeds $20\%$ when $k=80$. 
We observe that TM/ETM incurs higher cost than both $RanSwitch$ and $StaSwitch$ in most cases, especially for small privacy budgets. The reason is that TM/ETM has to introduce missing and empty values to fully satisfy the privacy guarantee. On the other hand, as we show the results in Table~\ref{table:total_cost}, when the privacy budget becomes even larger (e.g., $\epsilon>8$), TM/ETM cannot further benefit from it. This is because the utility gain of TM/ETM is capped at the upper limit of the threshold $k-1$, and is consistent with our analysis in Sec.~\ref{sec:introduction}, which motivates this work.

%On the other hand, TM/ETM incurs higher cost than both $RanSwitch$ and $StaSwitch$ in most cases, especially when the privacy budgets are small or large. The reason of the former is that TM/ETM has to introduce missing and empty values to fully satisfy the privacy guarantee, whereas the reason of the latter is that the utility gain of TM/ETM is capped at the upper limit of the threshold $k-1$. As such, we observe a flat line when the privacy budget is large, which means TM/ETM cannot further benefit from a larger privacy budget (e.g., $\epsilon\ge 8$ when $k=10$). This is consistent with our analysis in Sec.~\ref{sec:introduction}, which motivates this work.

%\begin{table}
%	\scriptsize
%	\caption{Total Cost of 3 Mechanisms of $k=10$}
%	\label{table:total_cost}
%	\begin{tabular}{|c|c|c|c|c|c|c|c|c|c|c|c|c|c|}
%		\hline
%		$\epsilon$ & 1 & 2 & 3 & 4 & 5 & 6 & 7 & 8 & 9 & 10 & 11 & 11 & 12   \\ \hline
%		TM/ETM & 13.1 & 7.6 & 5.1 & 4.0 & 3.0 & 2.0 & 2.0 & 1.0 & 1.0 & 1.0 & 1.0 & 1.0 & 1.0 \\ \hline
%		RanSwitch & 4.7 & 4.6 & 4.4 & 4.1 & 3.5 & 2.7 & 2.0 & 1.3 & 0.9 & 0.6 & 0.4 & 0.2 & 0.1 \\ \hline
%		StaSwitch & 4.1 & 4.0 & 3.8 & 3.5 & 3.0 & 2.4 & 1.8 & 1.2 & 0.8 & 0.5 & 0.3 & 0.2 & 0.1 \\ \hline
%	\end{tabular}
%\end{table}
\begin{table}
	\scriptsize
	\centering
	\caption{Total Cost of Larger Privacy Budgets ($k=10$)}
	\vspace{-0.08in}
	\label{table:total_cost}
	\begin{tabular}{|c|c|c|c|c|c|c|c|c|c|c|c|c|c|}
		\hline
		$\epsilon$ & 7 & 8 & 9 & 10 & 11 & 12 & 13 &14   \\ \hline
		TM/ETM &2.00 & 1.00 & 1.00 & 1.00 & 1.00 & 1.00 & 1.00 & 1.00\\ \hline
		RanSwitch & 1.96 & 1.34 & 0.89 & 0.56 & 0.35 & 0.22 & 0.13 & 0.08 \\ \hline
		StaSwitch & 1.77 & 1.24 & 0.83 & 0.54 & 0.34 & 0.21 & 0.12 & 0.08 \\ \hline
	\end{tabular}
\end{table}

\subsection{Utility Evaluation in Real Applications}

To compare the effectiveness of our proposed mechanisms $RanSwitch$ and $StaSwitch$ against TM/ETM, we measure their utilities in three real-world time series applications --- simple moving average, frequency counting, and trajectory clustering. We also compare them with {\it value perturbation} mechanisms in each application. For frequency counting, as each value is binary, we use Randomized Response (RR)~\cite{warner1965randomized} which achieves the best performance for binary data~\cite{wang2017locally}. For simple moving average and trajectory clustering, we use Piecewise mechanism (PM)~\cite{wang2019collecting}, the state-of-the-art LDP solution for numerical value perturbation.

\subsubsection{Simple Moving Average}
We conduct simple moving average of the stock's daily close price on {\it Stock-N} and calculate the mean square error (MSE) of the estimated results as
$
	\frac{1}{|S|-r+1} \sqrt{\sum\nolimits_{i=1}^{|S|-r+1}(m_i - m_i^\prime)^2},
$
where $m_i=\frac{1}{r}\sum_{j=i}^{i+r-1} S_j$ and $m_i^\prime=\frac{1}{r}\sum_{j=i}^{i+r-1} R_j$ are the moving averages from the original time series $S$ and the released one $R$, with averaging range $r$.

\begin{figure}[t]
	\centering
	\begin{minipage}{0.49\linewidth}
		\centerline{\includegraphics[width=\linewidth]{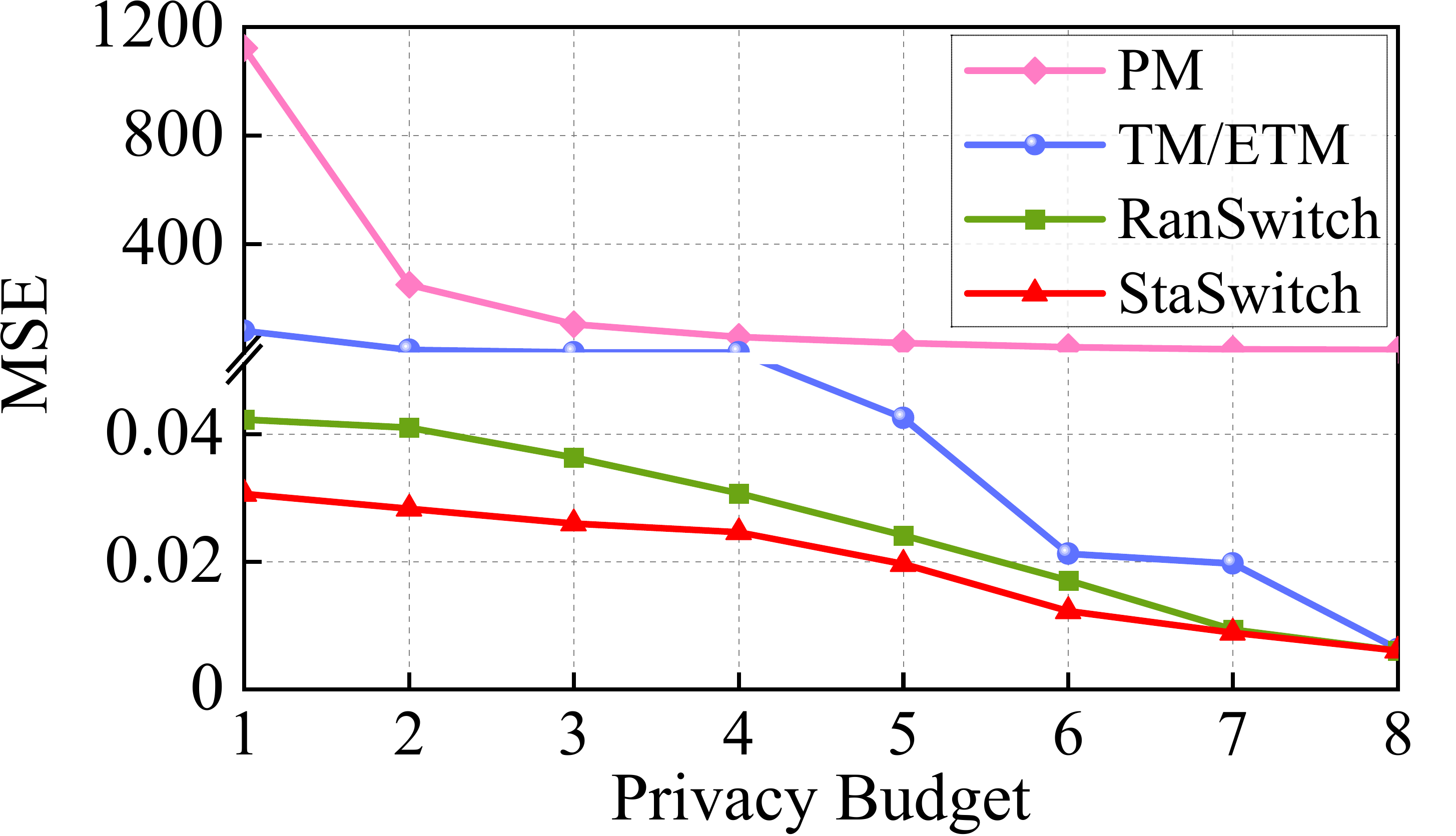}}
		\centerline{\footnotesize{(a) $k=10$, $r=10$}}
	\end{minipage}
	%\hspace{-0.1in}
%	\begin{minipage}{0.24\linewidth}
%		\centerline{\includegraphics[width=\linewidth]{figs/SMA_close_k10_range20-crop.pdf}}
%		\centerline{\footnotesize{(b) $k=10$, $r=20$}}
%	\end{minipage}
%	%\hspace{0.1in}
%	\begin{minipage}{0.24\linewidth}
%		\centerline{\includegraphics[width=\linewidth]{figs/SMA_close_k10_range30-crop.pdf}}
%		\centerline{\footnotesize{(c) $k=10$, $r=30$}}
%	\end{minipage}
	\begin{minipage}{0.49\linewidth}
		\centerline{\includegraphics[width=\linewidth]{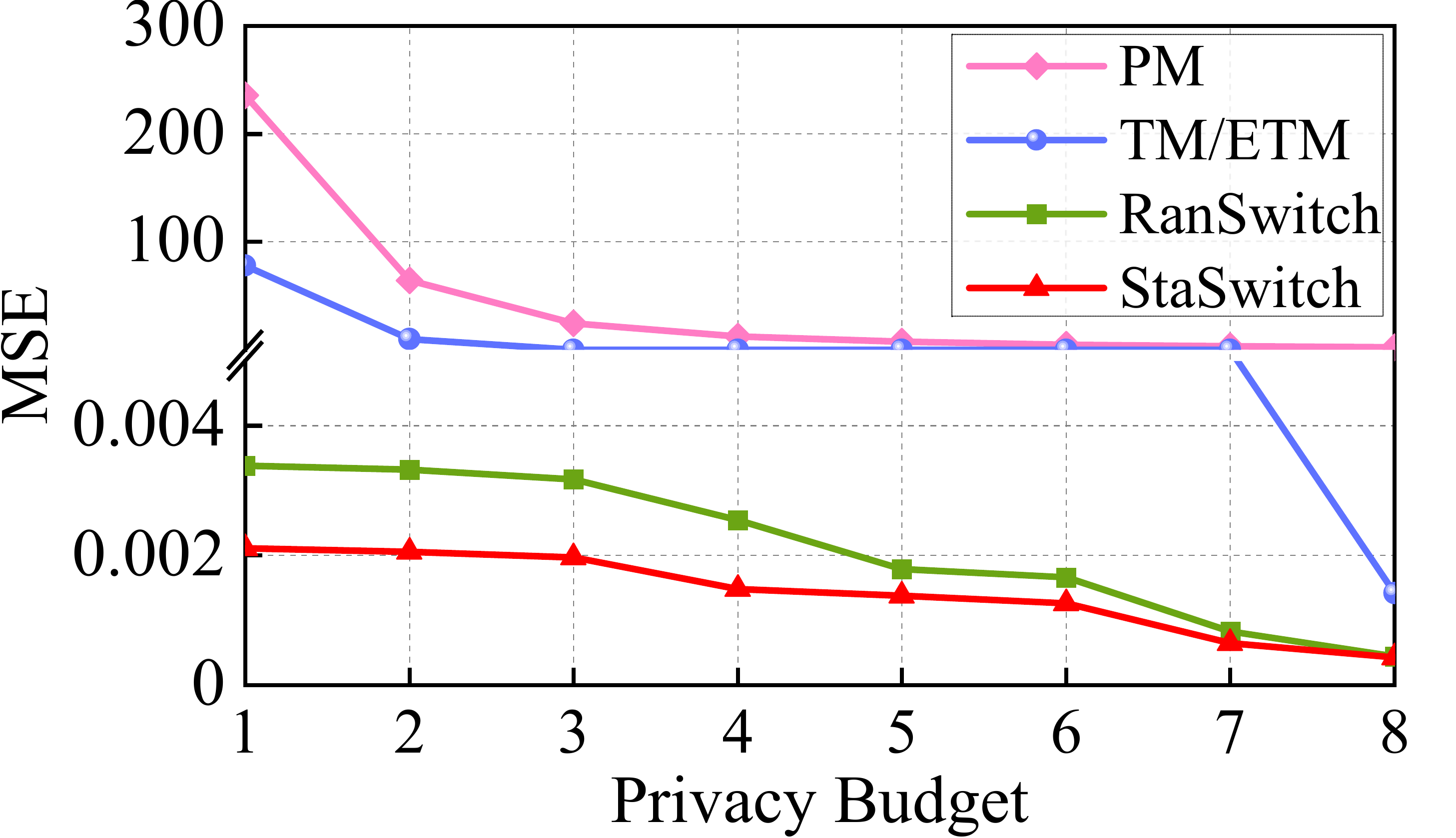}}
		\centerline{\footnotesize{(d) $k=10$, $r=40$}}
	\end{minipage}	
	\begin{minipage}{0.49\linewidth}
		\centerline{\includegraphics[width=\linewidth]{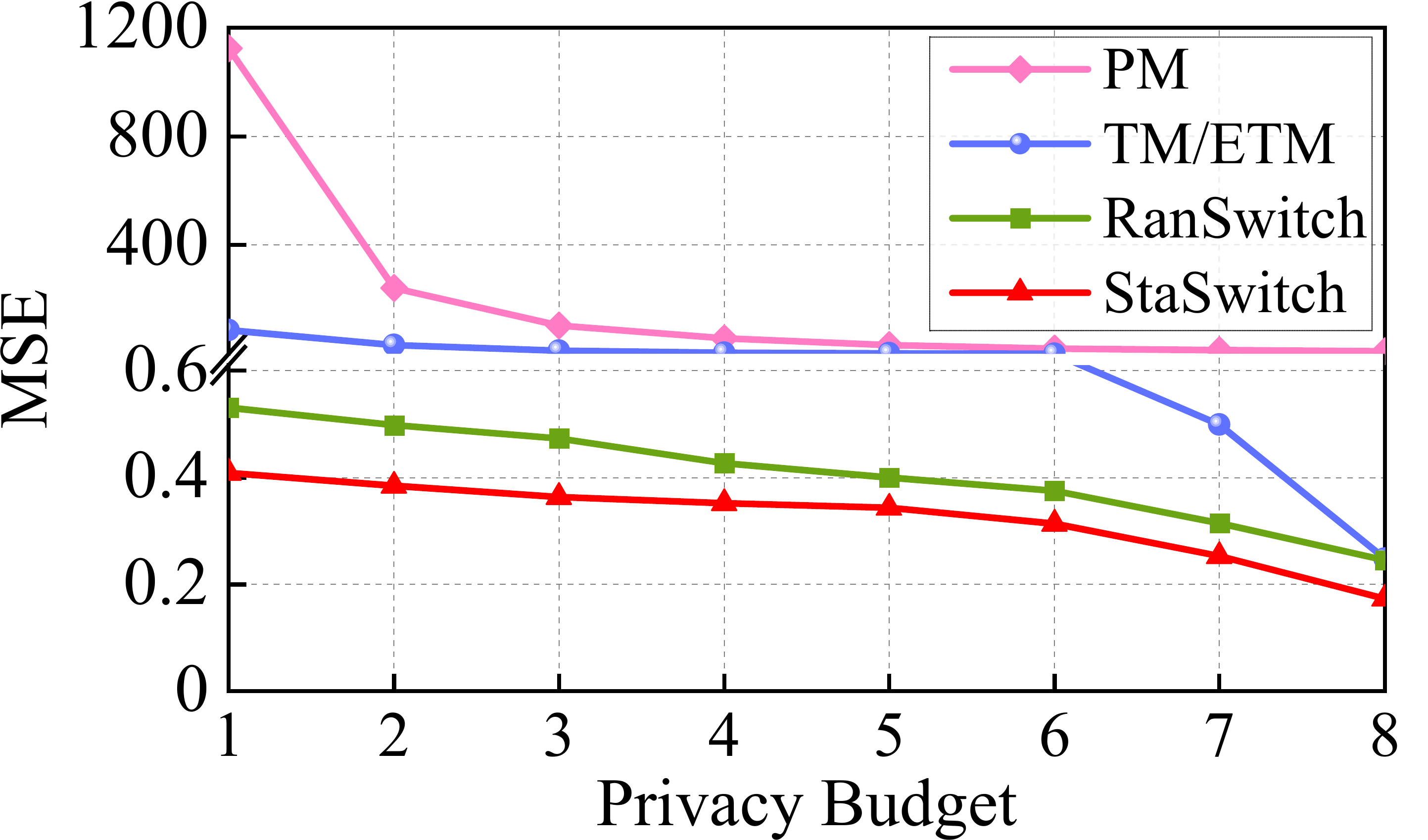}}
		\centerline{\footnotesize{(a) $k=40$, $r=10$}}
	\end{minipage}
	%\hspace{-0.1in}
%	\begin{minipage}{0.24\linewidth}
%		\centerline{\includegraphics[width=\linewidth]{figs/SMA_close_k40_range20-crop.pdf}}
%		\centerline{\footnotesize{(b) $k=40$, $r=20$}}
%	\end{minipage}
%	%\hspace{0.1in}
%	\begin{minipage}{0.24\linewidth}
%		\centerline{\includegraphics[width=\linewidth]{figs/SMA_close_k40_range30-crop.pdf}}
%		\centerline{\footnotesize{(c) $k=40$, $r=30$}}
%	\end{minipage}
	\begin{minipage}{0.49\linewidth}
		\centerline{\includegraphics[width=\linewidth]{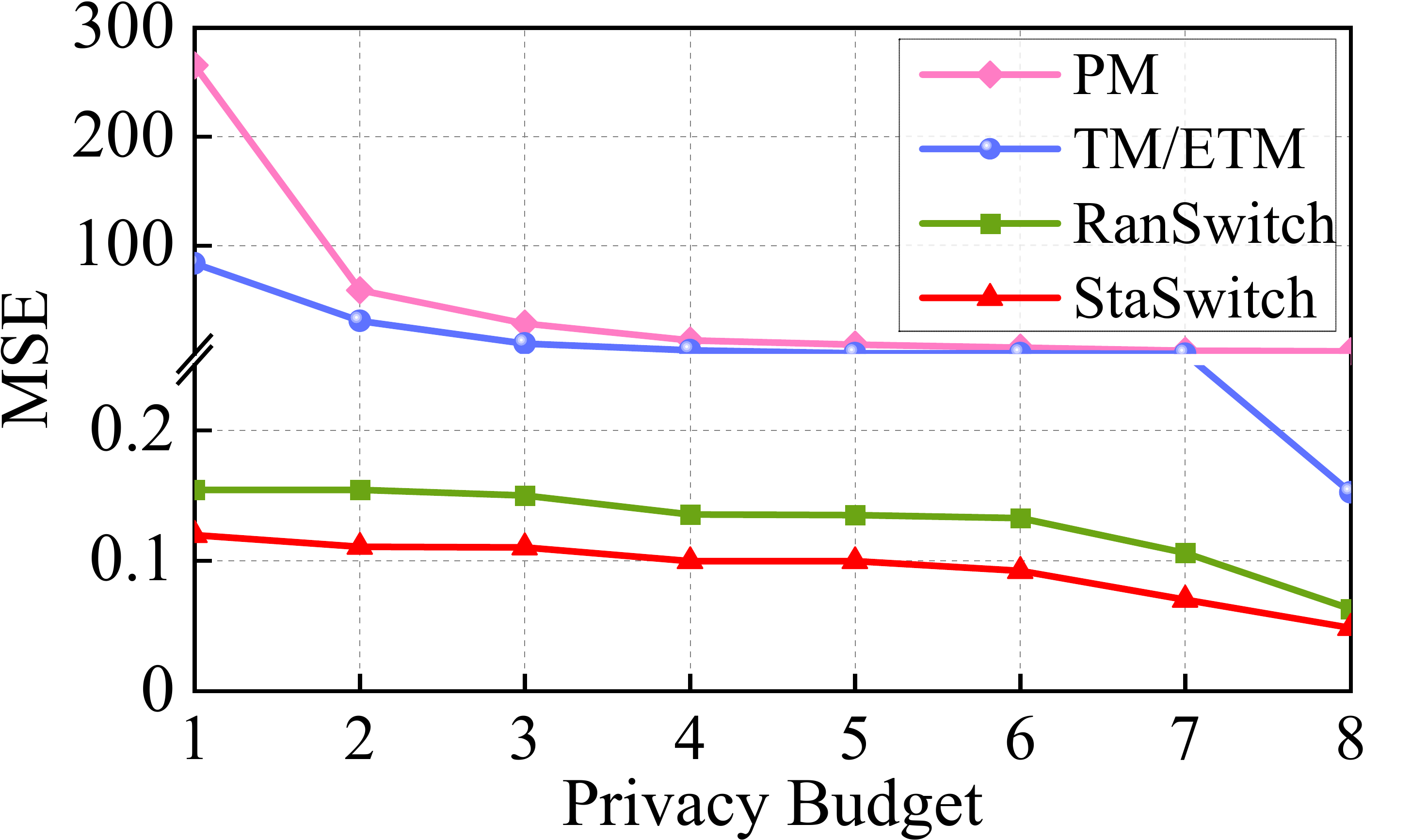}}
		\centerline{\footnotesize{(d) $k=40$, $r=40$}}
	\end{minipage}
	\vspace{-0.05in}
	\caption{Results of simple moving average on {\it Stock-N}}
	\vspace{-0.15in}
	\label{fig:SMA}
\end{figure}

Fig.~\ref{fig:SMA} plots the results, where the privacy budget $\epsilon$ varies from $1$ to $8$, and the length of sliding window $k$ and the averaging range $r$ are set to $10$ or $40$, respectively. 
To best accommodate all the results, we re-scale the y-axis in Fig.~\ref{fig:SMA}. 
Under different $k$ and $r$, PM always has the highest MSE, 1-3 orders of magnitude higher than that of TM/ETM and 3-5 orders of magnitude higher than that of $RanSwitch$ and $StaSwitch$. This demonstrates the superiority of {\it temporal perturbation} over {\it value perturbation} for simple moving average.
Both $RanSwitch$ and $StaSwitch$ outperform TM/ETM, and the gain is more eminent when privacy budget is small. This is because TM/ETM suffers from missing and empty values when the given privacy budget is relatively small, while $RanSwitch$ and $StaSwitch$ are free of missing or empty values thanks to the two-way nature of {\it switch} operation.
Between $RanSwitch$ and $StaSwitch$, $StaSwitch$ consistently outperforms $RanSwitch$ in all cases and the average gain exceeds $20\%$, as the stateful switch operation adopted by $StaSwitch$ further reduces the misalignment cost in the released time series.
In addition, when the averaging range $r$ increases from $10$ to $40$, the accuracy of all four mechanisms is improved. For PM, it is because more injected noise is canceled with each other according to law of large numbers; for TLDP mechanisms, it is because temporal perturbation is almost constrained within a sliding window of length $k$, but this disadvantage is mitigated when the averaging range becomes larger, e.g., $k=10$ and $r=40$. %With the increasing averaging range $r$, we also observe that $RanSwitch$ and $StaSwitch$ can benefit more than TM/ETM. %, as their performance gap becomes more significant with the increasing averaging range $r$.

\begin{figure*}[t]
	\centering
	\begin{minipage}{0.24\linewidth}
		\centerline{\includegraphics[width=\linewidth]{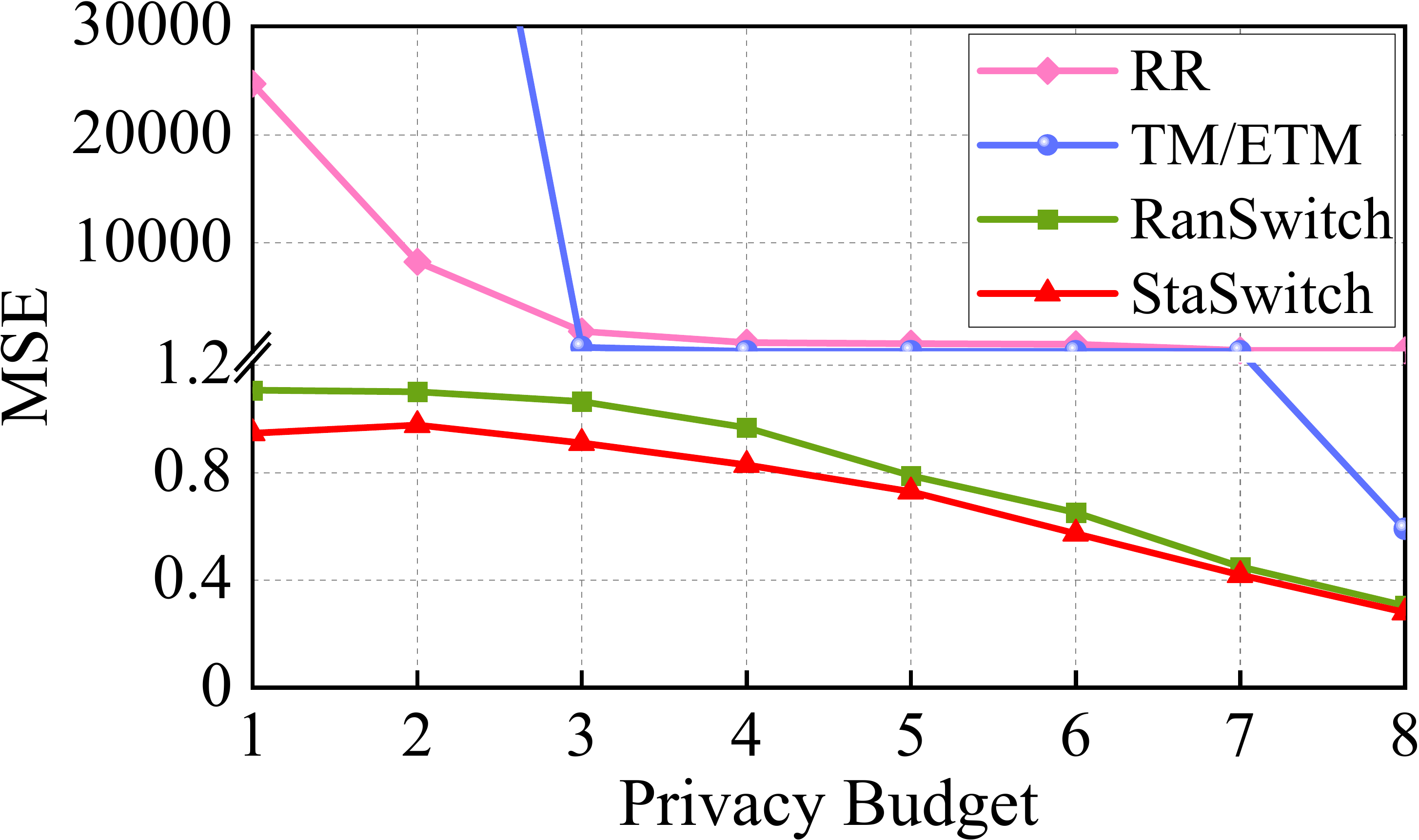}}
		\centerline{\footnotesize{(a) $k=10$}}
	\end{minipage}
	%\hspace{-0.1in}
	\begin{minipage}{0.24\linewidth}
		\centerline{\includegraphics[width=\linewidth]{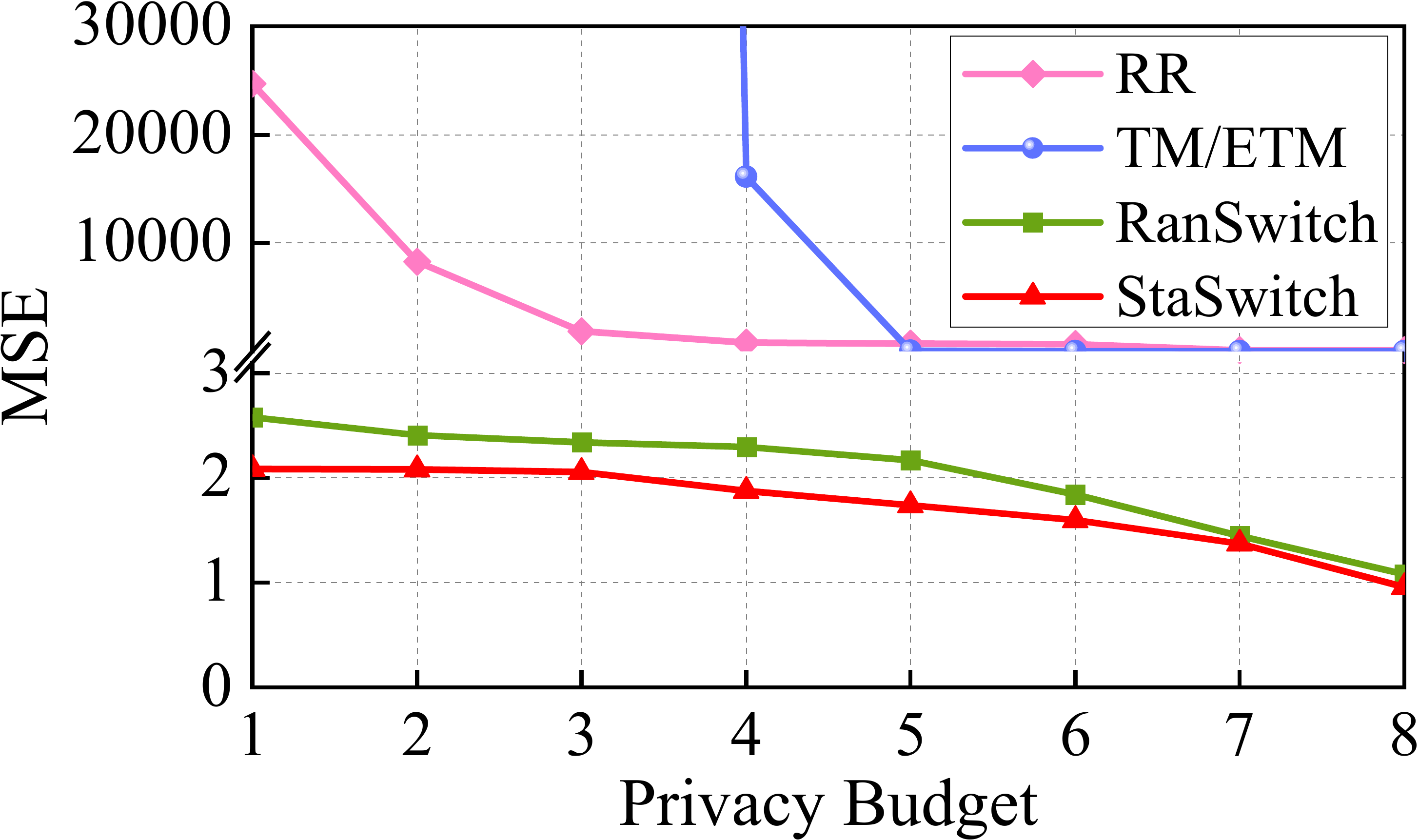}}
		\centerline{\footnotesize{(b) $k=20$}}
	\end{minipage}
	%\hspace{0.1in}
	\begin{minipage}{0.24\linewidth}
		\centerline{\includegraphics[width=\linewidth]{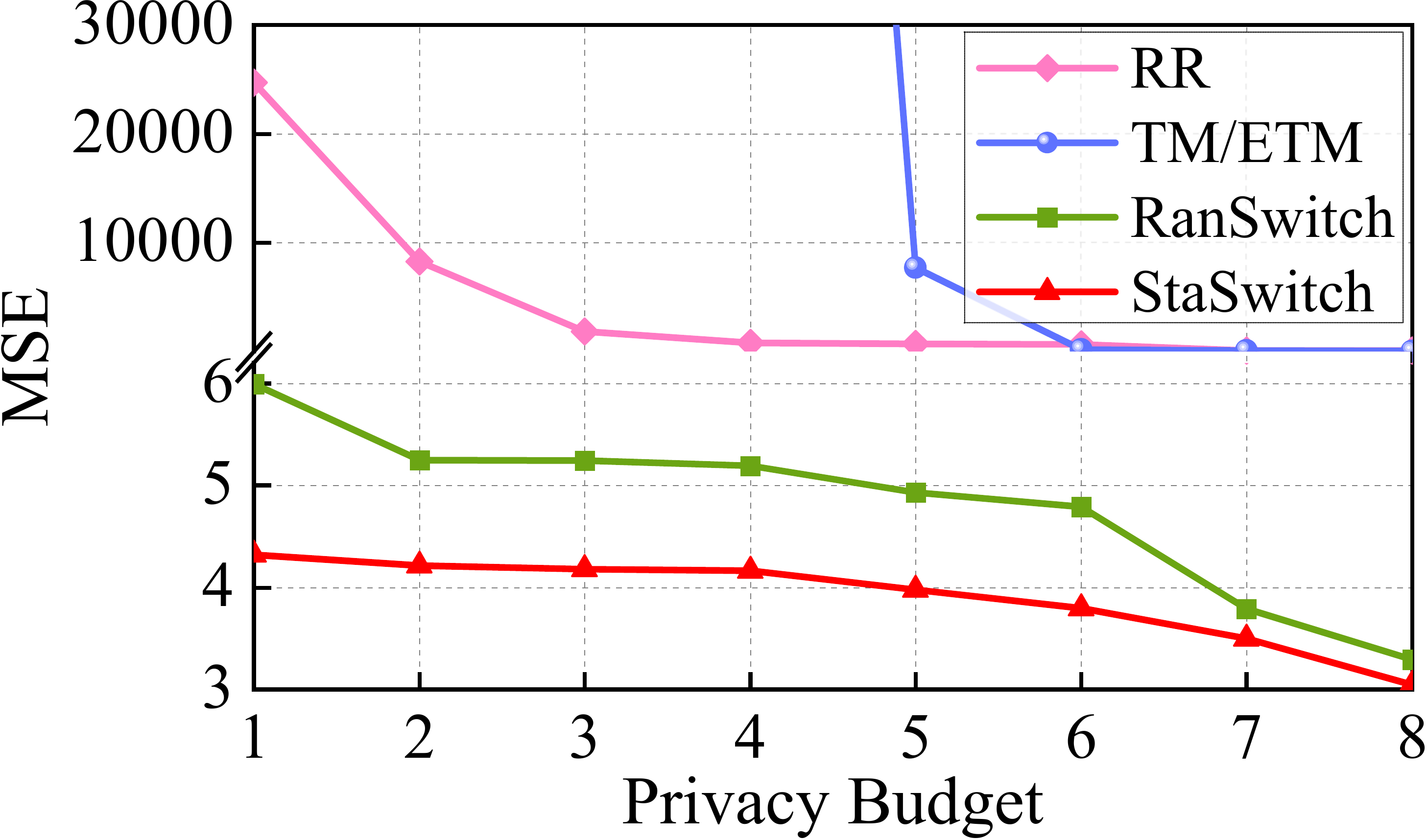}}
		\centerline{\footnotesize{(c) $k=40$}}
	\end{minipage}
	\begin{minipage}{0.24\linewidth}
		\centerline{\includegraphics[width=\linewidth]{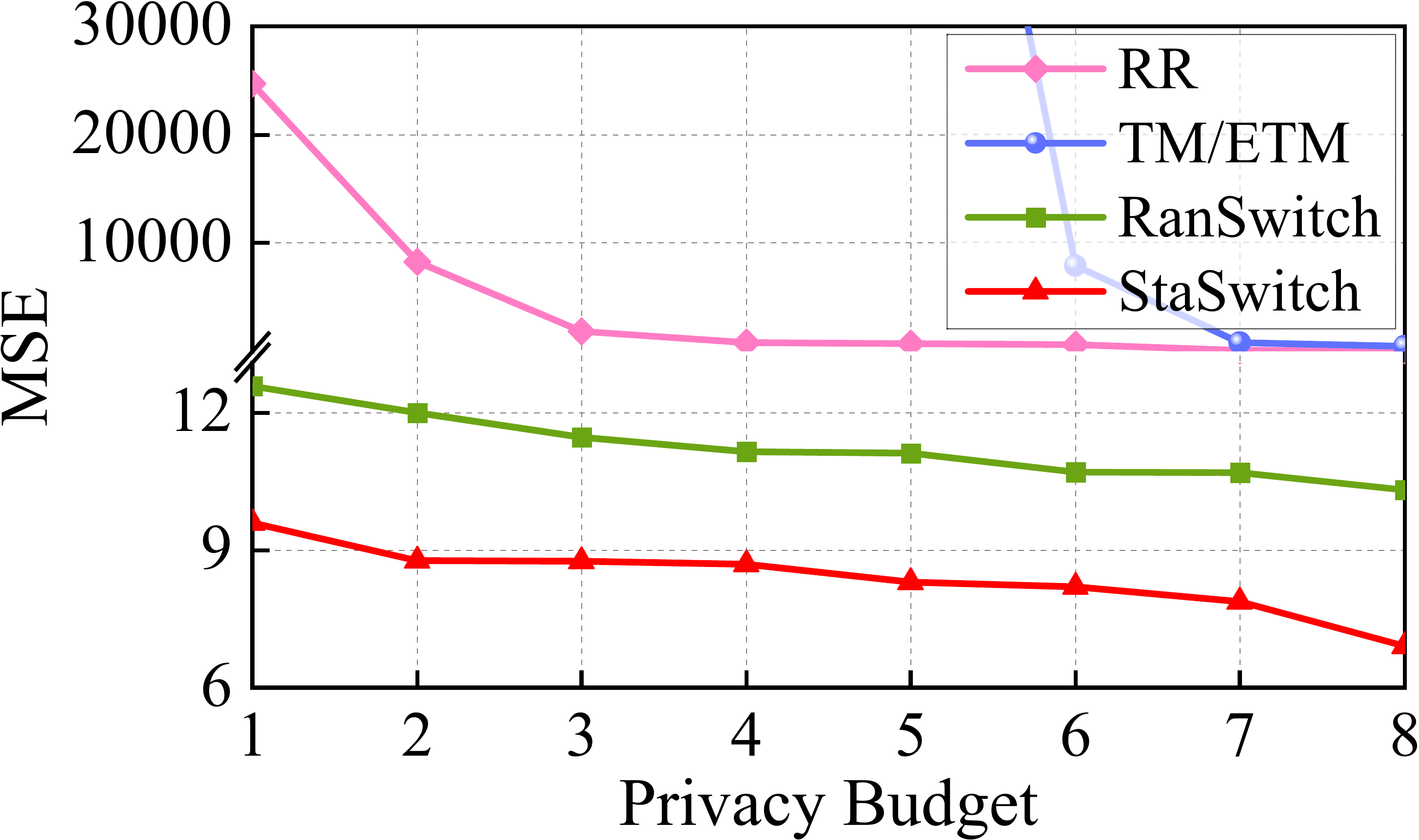}}
		\centerline{\footnotesize{(d) $k=80$}}
	\end{minipage}
	\vspace{-0.05in}
	\caption{Results of frequency counting on {\it Stock-B}}
	\vspace{-0.15in}
	\label{fig:FC}
\end{figure*}

\subsubsection{Frequency Counting}
To compare the accuracy of counting the frequency of value ``up'' on the dataset {\it Stock-B},  %At each timestamp, we calculate the true count since the first timestamp and use it as the ground truth.
we adopt RR, TM/ETM, $RanSwitch$ and $StaSwitch$ respectively to perturb the time series, and measure their deviation from the ground-truth count by calculating the MSE as
$
	\frac{1}{|S|}\sqrt{\sum\nolimits_{i=1}^{|S|} (c_i-c_i^\prime)^2},
$
where $c_i=\sum_{j=1}^{i}\mathbbm{1}(S_j=up)$ and $c_i^\prime=\sum_{j=1}^{i}\mathbbm{1}(R_j=up)$ are the ``up'' counts at timestamp $t_i$ in $S$ and $R$, respectively.

Fig.~\ref{fig:FC} plots the MSE of these mechanisms, where the time window length $k$ varies from $10$ to $80$ and the privacy budget $\epsilon$ varies from $1$ to $8$. Overall, $StaSwitch$ performs the best among the four mechanisms, followed by $RanSwitch$, and their gap becomes more eminent as the window length $k$ gets larger (e.g., $RanSwitch$ reduces $25\%$ of MSE when $k=80$). On the other hand, TM/ETM no longer outperforms RR (a {\it value perturbation} method) for small privacy budget  (e.g., $\epsilon \le 3$) or large sliding window (e.g., $k=80$). This is because frequency counting is very sensitive to missing values, as a missing of ``up'' decreases all subsequent counts by $1$. %Unfortunately, TM causes missing values easily for small $\epsilon$ or large $k$.

\subsubsection{Trajectory Clustering}
Same as~\cite{ye2021beyond}, we adopt a classic clustering algorithm, namely the $k$-medoids algorithm~\cite{park2009simple}, to cluster all 6,307 trajectories in the dataset {\it Trajectory} into $6$ groups. The clustering result over the original dataset is regarded as the ground truth. Then we perturb all trajectories by three TLDP mechanisms (i.e., $RanSwitch$, $StaSwitch$ and TM/ETM, $k=10$) and one {\it value perturbation} mechanism, namely PM, and apply the same clustering again.
To measure the similarity between the ground-truth clusters and those from perturbed trajectories, we adopt a classic metric, Normalized Mutual Information (NMI)~\cite{witten2016data}, whose range is $[0,1]$ and a larger NMI means more similarity. Table~\ref{table:trajectory_nmi_results} shows the results of four mechanisms with privacy budget from 1 to 8. We observe that PM has much smaller NMI than those TLDP mechanisms, because value perturbation heavily damages the utility of the released trajectories. On the other hand, $StaSwitch$ always achieves the highest NMI, followed by $RanSwitch$, which indicates they have more similar clustering results to the ground truth than TM/ETM.

\begin{table}
	\scriptsize
	\vspace{-0.05in}
	\caption{NMI on Trajectory Clustering}
	\vspace{-0.1in}
	\label{table:trajectory_nmi_results}
	\centering
	\begin{tabular}{|@{ }p{1.4cm}<{\centering}|@{ }p{0.55cm}<{\centering}|@{ }p{0.55cm}<{\centering}|@{ }p{0.55cm}<{\centering}|@{ }p{0.55cm}<{\centering}|@{ }p{0.55cm}<{\centering}|@{ }p{0.55cm}<{\centering}|@{ }p{0.55cm}<{\centering}|@{ }p{0.55cm}<{\centering}|}
		\hline
		\bf $\bm{\epsilon}$   & $1$   & $2$   & $3$   & $4$   & $5$   & $6$   & $7$   & $8$   \\ \hline		
		\bf PM         & $0.003$  & $0.005$  & $0.006$  & $0.007$  & $0.008$  & $0.008$  & $0.007$  & $0.008$ \\ \hline
		\bf TM/ETM  & $0.572$ & $0.590$ & $0.610$ & $0.616$ & $0.699$ & $0.705$ & $0.706$  & $0.768$ \\ \hline
		\bf RanSwitch  & $0.593$ & $0.613$ & $0.609$ & $0.623$ & $0.703$ & $0.713$ & $0.726$  & $0.845$ \\ \hline
		\bf StaSwitch  & $0.624$ & $0.650$ & $0.703$ & $0.725$ & $0.756$ & $0.745$ & $0.858$  & $0.928$ \\ \hline
	\end{tabular}
\vspace{-0.15in}
\end{table}

%% file: related.tex
In this section, we review existing works on differential privacy, with a focus on time series release.

{\bf Differential Privacy}.
Differential privacy was first proposed in the centralized setting~\cite{dwork2006calibrating}.
To avoid relying on a trusted data collector, local differential privacy (LDP) was proposed to let each user perturb her data locally~\cite{duchi2013local,li2019mobile}.
In the literature, many LDP techniques have been proposed for various statistical collection tasks, such as frequency of categorical data~\cite{erlingsson2014rappor, kairouz2014extremal, bassily2015local, wang2017locally}, and mean of numerical data~\cite{ding2017collecting, wang2019collecting,li2020estimating}. Recently, the research focus in LDP has been shifted to more complex tasks, such as itemset mining~\cite{wang2018locally}, marginal release~\cite{cormode2018practical, zhang2018calm}, graph data mining~\cite{qin2017generating, ye2020towards}, key-value data analysis~\cite{ye2019privkv, gu2020pckv, ye2021privkvm}, high-dimensional data analysis~\cite{du2021collecting, duan2022utility}, and learning problems~\cite{zheng2022protecting,lin2022towards}.

{\bf Centralized DP for Time Series}.
Existing works on centralized DP for time series focus on differentially private aggregate statistics, e.g., frequency estimation. Depending on the privacy requirement, a perturbation mechanism can satisfy \emph{event-level} privacy~\cite{dwork2010differential2, chen2017pegasus}, \emph{user-level} privacy~\cite{rastogi2010differentially, acs2014case}, or \emph{w-event} privacy~\cite{kellaris2014differentially, wang2016real}. There are also several strategies proposed to reduce the overall variance in the released statistics, such as Fourier transformation~\cite{rastogi2010differentially}, sampling~\cite{fan2013adaptive}, clustering~\cite{acs2014case}, and smoothing techniques~\cite{chen2017pegasus, ghayyur2018iot}. Another line of works also consider temporal correlation of continuously released time series data~\cite{xiao2015protecting, caoquantifying2019}.

{\bf LDP for Time Series}.
More recently, there are a number of studies on the problem of continual time series analysis under LDP. A technique based on \emph{memoization} was first proposed in the local setting~\cite{erlingsson2014rappor, ding2017collecting}. %, where each user first pre-computes and stores her randomized responses for all possible values. Then at each round of data collection, she responds with pre-computed responses without fresh randomization.
Besides that, Joseph \emph{et al.}~\cite{joseph2018local} design an approach to track changing statistics by assuming that user data are sampled from several evolving distributions. Erlingsson \emph{et al.}\cite{erlingsson2019amplification} further investigate a shuffle model for collecting correlated time series data. Wang \emph{et al.}~\cite{wang2021continuous} develop a framework for estimating the sum of real values over a time interval, and Bao \emph{et al.}~\cite{bao2021cgm} propose correlated Gaussian mechanism to reduce the noise injected to time series. Xue \emph{et al.}~\cite{xue2022ddrm} investigate continuous frequency estimation in the user population by exploring an optimal privacy budget allocation scheme to improve estimation accuracy.  

The above works are all based on \emph{value perturbation}. The most relevant work to this paper is~\cite{ye2021beyond}, which is the first work on TLDP privacy model and adopts \emph{temporal perturbation} to satisfy TLDP. %They proposed TM/ETM mechanism to \emph{dispatch} each incoming value to a subsequent timestamp. %Motivated by the prominent feature of temporal perturbation which preserves better utility over value perturbation,
However, this mechanism suffers from missing, repetition and empty cost, as well as limitations on settings of privacy parameters $\epsilon$ and $k$. These issues have been addressed in this paper.